\documentclass[runningheads]{llncs}
\usepackage[T1]{fontenc}
\usepackage{stmaryrd}

\usepackage{subcaption}
\usepackage{wrapfig}
\usepackage{algorithm}
\usepackage{algpseudocode}
\algnewcommand\algorithmicswitch{\textbf{switch}}
\algnewcommand\algorithmiccase{\textbf{case}}
\algdef{SE}[SWITCH]{Switch}{EndSwitch}[1]{\algorithmicswitch\ #1\ \textbf{do}}{\algorithmicend\ \algorithmicswitch}
\algdef{SE}[CASE]{Case}{EndCase}[1]{\algorithmiccase\ #1:}{\algorithmicend\ \algorithmiccase}
\usepackage{graphicx}
\usepackage{booktabs, multicol, multirow}
\usepackage{tikz}
\usepackage{amsmath,amssymb}
\usepackage{caption}
\definecolor{satgreen}{RGB}{76,153,76}
\definecolor{violred}{RGB}{204,60,60}
\definecolor{focusblue}{RGB}{50,100,180}
\definecolor{edgegray}{RGB}{140,140,140}
\definecolor{panelbg}{RGB}{248,248,248}
\definecolor{labelgray}{RGB}{80,80,80}

\usepackage{pifont}
\usepackage{array}

\usepackage{hyperref}
\usepackage{color}

\usepackage{amsmath, amssymb}

\usepackage{amsthm}
\usepackage{mathtools}          
\usepackage{bm}                 
\usepackage{xspace}             
\usepackage{xcolor}

\newcommand{\K}{\mathbb{K}}                        
\newcommand{\semiring}[1]{(\mathbb{K},\oplus,\otimes,\ominus,\bot,\top)}

\newcommand{\zero}{\bot}                           
\newcommand{\one}{\top}                            
\newcommand{\Bool}{\mathbb{B}}
\newcommand{\Rinf}{\mathbb{R}^{\infty}}

\newcommand{\Rpos}{\mathbb{R}^{+}}

\newcommand{\MinMax}{(\Rinf,\max,\min,-,{-\infty},{+\infty})}

\newcommand{\nodes}{\mathcal{V}}          
\newcommand{\timevar}{\mathbb{T}}         
\newcommand{\Ne}{\mathbb{N}}
\newcommand{\Be}{\mathbb{B}}
\newcommand{\mas}{\mathcal{M}}            

\newcommand{\agentstate}{\mathbf{s}}      
\newcommand{\jointstate}{\mathbf{X}}      
\newcommand{\Xc}{\mathcal{X}}             

\newcommand{\world}{w}                    
\newcommand{\Wc}{\mathcal{W}}             

\newcommand{\pos}{\mathbf{p}}

\newcommand{\mypara}[1]{\smallskip\noindent\textit{#1.}\enspace}

\newcommand{\graph}{\mathcal{G}}          
\newcommand{\graphs}{\mathbf{G}}         
\newcommand{\edge}{\mathcal{E}}           
\newcommand{\weights}{w}                  
\newcommand{\type}{\tau}            
\newcommand{\settypes}{\mathcal{T}}       

\newcommand{\InGO}[4]{\mathbf{In}_{#1,#2}^{#3,#4}}
\newcommand{\OutGO}[4]{\mathbf{Out}_{#1,#2}^{#3,#4}}
\newcommand{\InGOs}[3]{\mathbf{In}_{#1,#2}^{#3}}
\newcommand{\OutGOs}[3]{\mathbf{Out}_{#1,#2}^{#3}}
\newcommand{\In}{\mathbf{In}}
\newcommand{\Out}{\mathbf{Out}}
\newcommand{\FA}{\mathbf{FA}}
\newcommand{\EX}{\mathbf{EX}}

\newcommand{\Spmod}{\mathcal{S}}                   

\newcommand{\Nbrs}{\mathcal{N}}

\newcommand{\trace}{\sigma}

\newcommand{\lbl}{\nu}                             
\newcommand{\robt}[4]{\rho_{#1}(#2,#3,#4)}        

\newcommand{\Untilbd}[2]{\,\mathbf{U}_{[#1,#2]}\,}
\newcommand{\Glob}{\mathbf{G}}
\newcommand{\Globbd}[2]{\mathbf{G}_{[#1,#2]}}
\newcommand{\Ev}{\mathbf{F}}
\newcommand{\Evbd}[2]{\mathbf{F}_{[#1,#2]}}

\begin{document}
\title{An Algebraic Framework for Quantitative Semantics of Spatio-Temporal Logic with Graph Operators}
\titlerunning{Quantitative Semantics for STL-GO}
\author{Sheryl Paul\inst{1}\thanks{Equal contribution.} \and
Vidisha Kudalkar\inst{1}$^*$ \and
Anand Balakrishnan\inst{2}$^*$ \and
Tianhao Wu\inst{1} \and
Lars Lindemann\inst{3} \and
Jyotirmoy V. Deshmukh\inst{1}}
\authorrunning{Paul, Kudalkar \& Balakrishnan et.al.}
\institute{University of Southern California, Los Angeles, CA, USA\\
\email{\{sherylpa,kudalkar,wutianha,jdeshmuk\}@usc.edu} \and
University of Texas at Austin, Austin, TX, USA\\
\email{anandbal@utexas.edu} \and
ETH Z\"urich, Z\"urich, Switzerland\\
\email{llindemann@ethz.ch}}
%
\maketitle              %
\begin{abstract}
Spatio-Temporal Logic with Graph Operators (STL-GO) extends Signal Temporal Logic (STL) to multi-agent systems via graph operators that count neighboring agents satisfying a property, together with multi-agent quantifiers. 
While Boolean semantics for STL-GO are well-defined,
quantitative semantics have not yet been developed and 
existing quantitative semantics for spatio-temporal logics such
as STREL cannot capture the counting constraints in STL-GO's graph operators. We develop quantitative semantics for STL-GO as a layered algebraic
construction that separates temporal aggregation from graph-operator aggregation (governed
by an abstract accumulator with a monotone fold and readout).
We prove that soundness and completeness
reduce to monotonicity conditions on these components. We implement the framework and evaluate it on two multi-agent
environments: a 2D bounded region with stochastic
Dubins-car dynamics and a 3D Earth-satellite system,  under
four semantic instantiations (Boolean, min-max,
signed-deficit, and a hybrid), demonstrating the tradeoffs
between accumulator choices and reporting scalability in the
number of agents and time horizon.
\keywords{
  Signal Temporal Logic \and
  multi-agent systems \and
  robustness semantics \and
  graph operators \and
  semirings
}
\end{abstract}
\section{Introduction}
Signal Temporal Logic~\cite{stl-dejan} is widely used to describe objectives in many 
applications like in cyber-physical systems (CPS).
Its applications span robot task
objectives~\cite{silano2025stl,karagulle2025wstl}, control
synthesis~\cite{raman2014model}, 
scenario generation~\cite{kudalkar2024sampling}
and CPS monitoring~\cite{chen2022stl}. 
A key enabler of this success is the notion of
\emph{quantitative semantics}:  which assign a real-valued robustness score to each trajectory, based on \textit{how well} the trajectory satisfies the specification. A substantial body of work has explored such semantics, including classical robustness measures~\cite{fainekos2009robustness}, spatial and temporal robustness~\cite{donze2010robust,akazaki2015avstl,rodionova2016temporal}, stochastic robustness distributions~\cite{bartocci2013robustness}, and smooth approximations~\cite{mehdipour2019arithmetic,haghighi2019smooth,pant2017smooth}. These have enabled a wide range of applications, including robustness-aware monitoring~\cite{jaksic2018quantitative,Nenzi18_SSTL}, falsification~\cite{waga2020falsification}, model predictive control~\cite{raman2014model}, reinforcement learning~\cite{li2017reinforcement,erpo,formats_2024,anand_rl_1,anand_rl_2}, and signal classification~\cite{yoo2017rich}.

\noindent Modern systems, however, are generally multi-agent. Applications such as robot swarms, networked CPS, distributed sensing, and social systems require reasoning not only about temporal evolution but also about interactions between agents, naturally captured by time-varying graphs encoding communication, sensing, or influence relationships. Hence, specifying multi-agent behavior requires combining temporal logic with relational and spatial reasoning over graphs.

\noindent Several spatio-temporal logics have been proposed to address
this need, including
SSTL~\cite{Nenzi18_SSTL},
SpaTeL~\cite{spatel},
SaSTL~\cite{sastl},
STREL~\cite{bartocci2017monitoring,nenzi2022logic},
STL-GO~\cite{stlgo}, and
HyperLTL~\cite{finkbeiner2016first}.
These logics have enabled monitoring~\cite{balakrishnan2025hscc,Nenzi18_SSTL} and control synthesis~\cite{kudalkar2026sampling,alsalehi2021neural} of spatio-temporal cyber-physical systems.
 These logics differ in their treatment of spatial structure, expressiveness of quantification, and ability to handle dynamic relations. While Boolean semantics for these logics are generally well understood, quantitative semantics are less developed. Notably, STREL admits an elegant robustness interpretation based on min-max algebra~\cite{balakrishnan2025hscc}, while other frameworks either lack quantitative semantics or rely on constructions that are difficult to interpret or extend.
Among these, of particular interest to us is STL-GO, which significantly extends the expressiveness of prior frameworks by offering support for directed graph relations, multiple simultaneous graph layers, count-based constraints over neighbors, and quantification over both graphs and agents. However, this expressiveness introduces new challenges for quantitative semantics. In particular, graph operators in STL-GO require aggregating values over sets of neighbors subject to cardinality constraints. Such aggregation cannot be captured by standard semiring operations (e.g., min or max) alone. 

\noindent To illustrate the challenges, consider a team of agents
subject to the requirement that each agent must have at least
3 neighbors within communication range~$R$ satisfying a
property~$\varphi$ (Fig.~\ref{fig:motivation}).
In STL-GO, this is expressed as
$\mathbf{FA}_\mathcal{V}\;
\mathbf{In}^{\exists}_{\mathcal{G},[3,\infty)}\;\varphi$.
Plain STL lacks graph operators and multi-agent quantifiers;
expressing this would require expanding the counting
constraint into a conjunction over agents of disjunctions: \(
\bigwedge_{i \in V} \ \bigvee_{\{j_1,j_2,j_3\} \subseteq V \setminus \{i\}} 
\ \bigwedge_{k=1}^{3} \Big( \|x^{j_k} - x^{i}\| \le R \ \wedge \ \varphi(x^{j_k}) \Big),
\)
across all $\binom{N-1}{3}$ neighbor subsets with explicit
distance predicates.
Applying min-max robustness to this expansion yields nested
$\min$/$\max$ operations that collapse to the margin of the
3rd-best neighbor of the worst-off agent, with the following 
limitations:
(i)~\emph{Count-margin insensitivity}: whether a drone has exactly
three qualifying neighbors or five, the robustness is identical
so long as the third-best neighbor's margin is the same, even
though the latter configuration is far more resilient.
(ii)~\emph{Edge-margin insensitivity}: a drone just outside
communication range that robustly satisfies the
predicate is treated similarly to a drone that might be far away from the range, regardless
of how close it is to qualifying as a neighbor.
(iii)~\emph{Worst-case dominance}: at the system level, the
all-agent quantifier assigns robustness based on the single worst
agent, discarding the fact that the vast majority may be
performing well.
Moreover, STREL cannot express the counting requirement at all;
its spatial operators reason about path existence, not neighbor
counts.
{\setlength{\textfloatsep}{4pt plus 1pt minus 1pt}%
\captionsetup{belowskip=0pt}%
\begin{figure}[t]
  \centering
  \includegraphics[width=\textwidth]{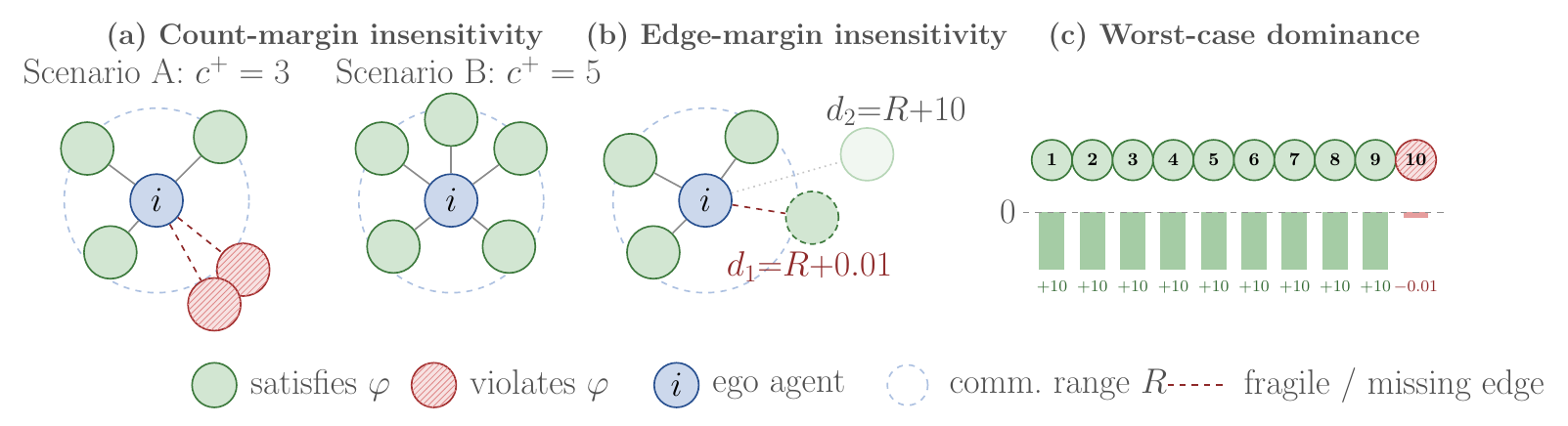}
\caption{\footnotesize Limitations of existing spatio-temporal robustness
    semantics.
    \textbf{(a)}~Both scenarios satisfy ``at least 3 neighbors
    satisfy $\varphi$,'' but A has zero spare neighbors and B has
    two; min-max robustness assigns both the same value.
    STREL cannot express this counting requirement.
    \textbf{(b)}~The dashed-border agent satisfies $\varphi$ but
    lies just outside range ($d_1{=}R{+}0.01$); the faded agent is
    far beyond $d_2 = R+10$.  Both are treated the same way by the robustness semantics.
    \textbf{(c)}~The all-agent STL-GO quantifier $\mathbf{FA}_\mathcal{V}$
    assigns robustness according to the worst agent ($-0.01$), discarding
    9 agents at $+10$.}
  \label{fig:motivation}
\end{figure}
These limitations motivate the development of richer quantitative
semantics for STL-GO that separately account for individual
signal margins, the number of satisfying neighbors, and the
robustness of the graph structure itself.

\noindent \textbf{Contributions.} 
Our contributions are:
(i)~a layered algebraic semantics for STL-GO that separates
temporal, graph, and multi-agent operators into distinct
algebraic components, with temporal and Boolean operators
functioning as in standard STL robustness;
(ii)~an abstract accumulator for graph operators that factors
into a monoid fold and a readout map, isolating
counting-specific aggregation and admitting product spaces
that propagate robustness values alongside qualifying agent
sets;
(iii)~multi-agent semantics in which agent quantifiers
aggregate over an agent-indexed valuation space via semiring
operations, completing the semantics for full STL-GO;
(iv)~soundness and completeness results showing that
monotone folds and readouts preserve threshold semantics,
with compositionality across all layers by structural
induction;
and (v)~experimental evaluation of four semantic
instantiations (Boolean, min-max, signed-deficit, and hybrid)
on two multi-agent environments, demonstrating accumulator
tradeoffs and scalability in agents and time horizon.
\section{Preliminaries}
We first describe the multi-agent system and the spatio-temporal logic STL-GO.
\subsection{Multi-Agent System}
We consider a multi-agent system consisting of a finite set of agents denoted by $\nodes = \{1, \dots, N\}.$
 Each agent $i \in \nodes$ has a state space $\Xc$, and the joint system
state at time $t$ is denoted by
\(
\jointstate_t = (x_t^1, \dots, x_t^N) \in \Xc^{|\nodes|},
\)
where $x_t^i \in \Xc$ is the state of agent $i$ at time $t$. In addition to agent states, the system may include an external
environment or world state $\world_t \in \Wc$. Here, $\Wc$ represents the environment space. Agent interactions are modeled by a collection of time-varying,
weighted directed graphs
\(
\graphs_t = \{\graph_t^{\type} \mid \type \in \settypes\},
\)
where $\settypes$ denotes the set of all graph types, and each graph $\graph_t^{\type} = (\nodes, \edge_t^{\type},
\weights_t^{\type})$ consists of: (i) a set of directed edges $\edge_t^{\type} \subseteq \nodes \times \nodes$, (ii)  a weight function $\weights_t^{\type} : \edge_t^{\type} \to \mathbb{R}$. Each graph type $\type \in \settypes$ represents a distinct interaction
modality. A \emph{multi-agent execution trace} is a sequence
\(
\mas = \{(\jointstate_0, \world_0, \graphs_0),
(\jointstate_1, \world_1, \graphs_1), \dots \}
\)
over discrete time steps $t \in \mathbb{N}$.
We write $(\mas,i,t)$ to denote evaluation of agent-local formulas at
agent $i \in \nodes$ and time $t$, and $(\mas,t)$ for multi-agent
formulas evaluated over the joint system at time $t$.
We write $(\mas,i,t) \models \varphi$ to denote that an agent-local formula $\varphi$ is \emph{Boolean satisfied} at agent $i$ and time $t$ and $(\mas,t) \models \phi$ to denote that a multi-agent formula $\phi$ is \emph{satisfied} over the system.\footnote{We only consider traces over discrete-time in this work.}
 \subsection{Signal Temporal Logic with Graph Operators (STL-GO)}
\label{sec:stlgo}
\emph{Spatio-Temporal Logic with Graph Operators} (STL-GO)~\cite{stlgo}
extends STL~\cite{stl-dejan} to reason simultaneously about
spatio-temporal agent behavior and the \emph{topological} relationships
between agents encoded in interaction graphs.
Its syntax and semantics are stratified into agent-local formulas $\varphi$
and multi-agent formulas $\phi$.

\noindent\textbf{Agent-Local Formulas.}
The syntax for agent-local STL-GO formulas is:
\[
  \varphi \;::=\; \top \;\mid\; \mu_x \;\mid\; \neg\varphi \;\mid\;
  \varphi \wedge \varphi \;\mid\;
  \varphi\,\Untilbd{t_1}{t_2}\,\varphi \;\mid\;
  \InGO{\graph}{E}{W}{\#}\varphi \;\mid\;
  \OutGO{\graph}{E}{W}{\#}\varphi,
\]
where $\mu_x : \Xc \to \Be$ is an atomic predicate over a single agent's
state;
$\neg$ and $\wedge$ are the standard Boolean connectives;
$\Untilbd{t_1}{t_2}$ is the bounded Until operator from STL; and
$\InGO{\graph}{E}{W}{\#}$, $\OutGO{\graph}{E}{W}{\#}$ are the \emph{incoming}
and \emph{outgoing graph operators} introduced by STL-GO, where $W =
[w_1, w_2] \subseteq \mathbb{R}^{\infty}$ is an edge-weight interval, $E = [e_1,
e_2]$ with $e_1 \in \Ne,\ e_2 \in \Ne \cup \{+\infty\}$ constrains the count of qualifying edges, and $\# \in
\{\exists, \forall\}$ quantifies over the graph collection $\graphs_t$.
The semantics $(\mas, i, t) \models \varphi$ are defined inductively;
Boolean and temporal clauses follow STL:
\begin{align*}
  (\mas,i,t) &\models \mu_x
    &&\text{iff } \mu_x(\agentstate^i_t), \\
  (\mas,i,t) &\models \neg\varphi
    &&\text{iff } (\mas,i,t) \not\models \varphi, \\
  (\mas,i,t) &\models \varphi_1 \wedge \varphi_2
    &&\text{iff } (\mas,i,t)\models\varphi_1
       \text{ and } (\mas,i,t)\models\varphi_2, \\
  (\mas,i,t) &\models \varphi_1\,\Untilbd{t_1}{t_2}\,\varphi_2
    &&\text{iff } \exists t' \in t\oplus[t_1,t_2].\;
       (\mas,i,t')\models\varphi_2 \\
    &&&\phantom{\text{iff }}\land\;
       \forall t'' \in [t,t'].\; (\mas,i,t'')\models\varphi_1.
\end{align*}
The graph operators count the neighbors of agent $i$ that satisfy 
$\varphi$ and lie within the weight interval $W$:
\begin{align*}
  \scalebox{0.92}{$\displaystyle
  (\mas,i,t) \models \InGO{\graph}{E}{W}{\#}\varphi
  \quad\text{iff}\quad
  \#\,\graph^{\type} \in \graphs_t.\;
  \Bigl|\bigl\{(j,i)\in\edge_t^{\type}
    \mid \weights_t^{\type}(j,i)\in W,\;
    (\mas,j,t)\models\varphi\bigr\}\Bigr| \in E
  $}
\end{align*}
The $\Out$ operator is similarly defined except using $(i,j)\in\edge_t^{\type}$ instead.
Here $\# = \exists$ requires the count condition to hold in \emph{some}
graph $\graph^{\type} \in \graphs_t$; $\# = \forall$ requires it in
\emph{every} graph.
When edge weights are unconstrained, we abbreviate
$\InGOs{\graph}{E}{\#}\varphi$ and $\OutGOs{\graph}{E}{\#}\varphi$
(setting $W = (-\infty,\infty)$).

\noindent \textbf{Multi-Agent Formulas.}
Multi-agent STL-GO formulas $\phi$ compose agent-local formulas and
predicates over the joint system state:
\[
  \phi \;::=\; \top \;\mid\; \mu \;\mid\; i.\varphi \;\mid\;
  \neg\phi \;\mid\; \phi\wedge\phi \;\mid\;
  \phi\,\Untilbd{t_1}{t_2}\,\phi,
\]
where $\mu : \Xc^{|\nodes|} \times \Wc \to \Be$ is an atomic predicate
over the joint agent and world state, and $i.\varphi$ embeds the
agent-local formula $\varphi$ for agent $i$ into the multi-agent context.
The semantics $(\mas, t) \models \phi$ follow the same structure as the
agent-local case, with:
\(
  (\mas,t) \models \mu
    \text{ iff } \mu(\jointstate_t, \world_t), \text{ and }
  (\mas,t) \models i.\varphi
    \text{ iff } (\mas,i,t)\models\varphi.
\)
All remaining clauses ($\neg$, $\wedge$, $\Untilbd{t_1}{t_2}$) are
identical in structure to the agent-local case.
The derived operators $\Evbd{t_1}{t_2}$ and $\Globbd{t_1}{t_2}$
follow from Until as in STL.  For convenience, the agent-universal and
agent-existential shorthands are:
\(
  \FA_\nodes\,\varphi \;\triangleq\; \bigwedge_{i\in\nodes} i.\varphi,
  \text{ and }
  \EX_\nodes\,\varphi \;\triangleq\; \bigvee_{i\in\nodes} i.\varphi.
\)

\noindent The semantics of STL-GO presented so far assume the standard Boolean 
satisfaction semantics; i.e., for a given agent-centric STL-GO formula 
$\varphi$, or for a multi-agent STL-GO formula $\phi$, a given multi-agent
execution trace either satisfies the given formulas or violates them.
{\em Quantitative semantics} generalize the notion of Boolean semantics by
defining a function that maps a formula and a multi-agent
trace to some (signed) number that indicates a quantitative degree of satisfaction. Many different quantitative
semantics for logics such as STL and STREL have been defined
in the literature \cite{fainekos2009robustness,donze2010robust,rodionova2016temporal,balakrishnan2025hscc}. Of particular relevance to this paper are generalization of quantitative semantics as operations defined using suitable algebras. We also define the quantitative semantics of STL-GO more generally using specific algebraic operations. Before we define them in the next section, we review some key ideas from algebra.

\subsection{Background on Algebraic Structures}
\begin{definition}[Monoid]
A \emph{monoid} is a triple $(M,\oplus,e)$ consisting of a set $M$, a binary operation
\(\oplus : M \times M \to M,\)
and an element $e \in M$, such that for all $a,b,c \in M$:
    \(
    (a \oplus b)\oplus c = a \oplus (b \oplus c)
    \)
     and
    \(
    a \oplus e = e \oplus a = a
    \), i.e., all the elements satisfy the property of associativity and $e$ acts as an identity element.
Further, a \emph{commutative monoid} is a monoid $(M,\oplus,e)$ that additionally satisfies, for all $a,b \in M$:
\(a \oplus b = b \oplus a.
\)
\label{def:monoid}
\end{definition}
\begin{definition}[Semiring~\cite{golan1999semirings,kuich1986semirings}]
\label{def:semiring}
A tuple $\K = (K, \oplus, \otimes, \zero, \one)$ is a \emph{semiring}
with underlying set $K$ if
$(K, \oplus, \zero)$ is a commutative monoid with identity $\zero$;
$(K, \otimes, \one)$ is a monoid with identity $\one$;
$\otimes$ distributes over $\oplus$;
and $\zero$ is an annihilator for $\otimes$
(i.e.\ $k \otimes \zero = \zero \otimes k = \zero$ for all $k \in K$).
 \end{definition}
\noindent A semiring $\K$ is \emph{commutative} if $\otimes$ is also commutative,
and \emph{simple} if $k \oplus \one = \one$ for all $k \in K$.
It is \emph{additively} (resp.\ \emph{multiplicatively}) \emph{idempotent}
if $k \oplus k = k$ (resp.\ $k \otimes k = k$) for all $k \in K$.
A semiring that is commutative, additively idempotent, and simple is
called a \emph{constraint semiring} (c-semiring).
A c-semiring that is also multiplicatively idempotent induces a
bounded distributive lattice, with $\one$ as supremum and $\zero$ as infimum.
\begin{example}
Consider the semiring $(\mathbb{R} \cup \{-\infty, +\infty\}, \max, \min, -\infty, +\infty)$.
Then $(\mathbb{R}, \max, -\infty)$ is a commutative monoid with identity $-\infty$,
and $(\mathbb{R}, \min, +\infty)$ is a monoid with identity $+\infty$.
By the distributivity of $\min$  over $\max$,
we have $\min(a,\max(b,c))$ $= \max(\min(a,b), \min(a,c))$.
The element $-\infty$ is an annihilator for $\min$, since
\(
\min(a,-\infty) = -\infty.
\)
Both operations are idempotent, i.e., $\max(a,a)=a$ and $\min(a,a)=a$.
The semiring is commutative, and simple since $\max(k,+\infty)=+\infty$
for all $k \in K$.
\end{example}
\noindent When considering a unary operator, \(\ominus: K \to K\), on a bounded distributive lattice \(K\), 
these operators give rise to a special case of a useful algebra:
\begin{definition}[De Morgan Algebra~\cite{cignoli1983dualities}]
\label{def:deMorgan}
A \emph{De Morgan algebra} is a structure
$(K, \oplus, \otimes, \ominus, \zero, \one)$ such that
$(K, \oplus, \otimes, \zero, \one)$ is a bounded distributive lattice and
$\ominus$ is a negation on $K$ satisfying, for all $a, b \in K$: $\ominus(a \oplus b) = \ominus a \otimes \ominus b, \ 
  \ominus(a \otimes b) = \ominus a \oplus \ominus b,  \ 
  \ominus\ominus a = a.$
\end{definition}
\noindent The two De Morgan algebras of primary interest here are:
(i) The \emph{Boolean algebra} $(\Bool, \vee, \wedge, \neg, \mathbf{0}, \mathbf{1})$,
        giving qualitative (Boolean) satisfaction semantics.
(ii) The \emph{min-max algebra} $(\Rinf, \max, \min, -, {-\infty}, {+\infty})$
        over the extended reals, which recovers the quantitative (robustness)
        semantics of STL~\cite{fainekos2009robustness}.
\noindent Since we are interested in the semantics of multi-agent systems,
we will consider how such semirings can be extended to multiple components, and
how we can generalize operations on multiple semirings.
Let semirings be $\{(K_\ell,\oplus_\ell,\otimes_\ell,\zero_\ell,\one_\ell)\}_{\ell=1}^m$ and their product semiring is $K := \prod_{\ell=1}^m K_\ell$, with
componentwise operations:
$ (x \oplus y)_\ell := x_\ell \oplus_\ell y_\ell, \ (x \otimes y)_\ell := x_\ell
\otimes_\ell y_\ell$
and identities:
\( \zero := (\zero_\ell)_\ell, \  \one := (\one_\ell)_\ell. \)
From this, we can algebraically define the notion of a general \emph{commutative
monoid} operating on this product semiring as a \emph{semimodule}.
\begin{definition}[Semimodule]
Let $(K,\oplus,\otimes,\zero,\one)$ be a semiring.
A semimodule over $K$ is a commutative monoid $(S,\oplus_S,\zero_S)$ together with a scalar multiplication
$ \otimes : K \times S \to S$,
such that for all $k,k' \in K$ and $s,s' \in S$:
\begin{align*}
    k \otimes (s \oplus_S s') &= (k \otimes s) \oplus_S (k \otimes s'); \quad &&(k \oplus k') \otimes s = (k \otimes s) \oplus_S (k' \otimes s) \\
    (k \otimes k') \otimes s &= k \otimes (k' \otimes s); \quad &&\one \otimes s = s; \quad k \otimes \zero_S = \zero_S 
\end{align*}
\end{definition}

\begin{definition}[$\K$-Labeled Trace]
\label{def:trace}
Let $(K, \oplus, \otimes, \ominus, \zero, \one)$ be a De Morgan
algebra and $L$ a space universe.  A \emph{$\K$-labeled trace} is a
pair $(\trace, \lbl)$ where $\trace = \Spmod_0\Spmod_1\cdots$ is a
finite sequence of spatial models and
$\lbl(\Spmod_i, l, \mu) \in K$ is the value of atomic predicate $\mu$
at location $l \in L$ in model $\Spmod_i$.
\end{definition}

\section{Layered Algebraic Semantics for STL-GO}
In addition to the operators introduced by STL, STL-GO also introduces \emph{graph operators} and \emph{multi-agent operators} which require mechanisms beyond those required for defining standard STL robustness. While the latter includes only Boolean and temporal operators that can be governed by a single algebraic structure such as a De Morgan algebra or semiring, STL-GO requires multiple forms of algebraic structures. To capture this uniformly, we present STL-GO semantics as a
layered algebraic construction.
We begin with the temporal semantics, which provide the base algebraic
structure.
\subsection{Semirings for Temporal Operators}
We begin with ego-centric temporal operators for STL-GO which generalize robustness semantics of STL. 
Let $(K,\oplus,\otimes,\ominus,\bot,\top)$
be an algebra,
where
 $(K,\oplus,\bot)$ is a commutative monoid,
$(K,\otimes,\top)$ is a monoid,
$\ominus$ is an involutive negation satisfying De Morgan laws
when present.
The set $K$ represents the space of semantic values assigned to formulas.

\noindent\textbf{Ego-centric boolean and temporal semantics.} To evaluate a formula $\varphi$
for agent~$i$ at time~$t$, the ego-centric semantics $\robt{i}{\trace}{\varphi}{t} \in K$ can be defined as:
\begin{align*}
  \robt{i}{\trace}{\mu}{t}
    &= \lbl(\Spmod_t, i, \mu), \quad
  \robt{i}{\trace}{\varphi_1 \wedge \varphi_2}{t}
    = \robt{i}{\trace}{\varphi_1}{t}
       \otimes
       \robt{i}{\trace}{\varphi_2}{t}, \notag\\
  \robt{i}{\trace}{\neg\varphi}{t}
    &= \ominus\,\robt{i}{\trace}{\varphi}{t}, \quad
  \robt{i}{\trace}{\varphi_1 \vee \varphi_2}{t}
    = \robt{i}{\trace}{\varphi_1}{t}
       \oplus
       \robt{i}{\trace}{\varphi_2}{t}, \notag\\
  \robt{i}{\trace}{\varphi_1 &\Untilbd{t_1}{t_2} \varphi_2}{t}
    =
    \bigoplus_{t' \in t+[t_1,t_2]}
    \biggl(
      \robt{i}{\trace}{\varphi_2}{t'}
      \otimes
      \bigotimes_{t'' \in [t,\,t']}
      \robt{i}{\trace}{\varphi_1}{t''}
    \biggr)
  \label{eq:stlgo-temporal}
\end{align*}
Derived operators such as $\Ev$ and $\Glob$ follow in the usual way.
\noindent \begin{example}
    Different choices of $(K,\oplus,\otimes)$ yield different
interpretations: \\ \textit{(i) Boolean Semantics:}  $K=\{0,1\}$, $\oplus=\vee$, $\otimes=\wedge$ recovers the classical qualitative satisfaction.  \textit{(ii) Min-max Semantics:} $K=\Rinf$, $\oplus=\max$, $\otimes=\min$ recovers the standard STL robustness, where values represent
    perturbation margins.
\end{example}

\noindent Recall the definition of product semirings. This construction allows formulas to simultaneously produce multiple semantic quantities. E.g., we may combine a robustness value with quantities such as counts, or relational data. The temporal semantics may be extended componentwise to this setting, ensuring consistency across outputs. This multi-output structure will be essential for defining graph operators which we describe below.
\subsection{Multi-agent and Graph Operators}
\label{subsec:graph-operators}
\textbf{Graph Operators. }
We now define the semantics of graph operators. Unlike temporal
operators, which are induced directly by the semiring structure on $K$,
graph operators require aggregation over dynamically defined
neighborhoods of agents. This aggregation is not, in general, expressible
purely in terms of the semiring operations $(\oplus,\otimes)$.

\noindent Let $\nodes$ denote the set of agents, and define $S := K^{\nodes}$.
Each element $s \in S$ assigns to every agent $i \in \nodes$
a semantic value $s(i) \in K$. The space $S$ inherits algebraic structure from $K$:
 $(S,\oplus,\bot)$ is a commutative monoid under addition: $(s \oplus s')(i) := s(i) \oplus s'(i)$,
 $S$ is a semimodule over $K$ with scalar multiplication defined: $(k \otimes s)(i) := k \otimes s(i)$. For a fixed trace $\trace$, formula $\varphi$, and time $t$, the
ego-centric semantics induce an element
$s_{\varphi,t} \in S, \ 
s_{\varphi,t}(i) := \robt{i}{\trace}{\varphi}{t}$,
which collects the local evaluations across all agents.

\mypara{Eligible Neighbors}
For an agent $i \in \nodes$, time $t$, graph type
$\type \in \settypes$, and weight interval $W$, we define the sets of eligible neighbors:
$  \Nbrs_{i}^{\type,\In}(t,W)
  =
  \{\, j \in \nodes
    \mid (j,i)\in \edge_t^{\type},\;
          \weights_t^{\type}(j,i)\in W \,\}$,
$\Nbrs_{i}^{\type,\Out}(t,W)$ is defined similarly except with $(i,j)\in \edge_t^{\type}$. The graph operators act by restricting the global valuation $s_{\varphi,t} \in S$ to the neighborhood of $i$, yielding a multiset of values: \( \{ s_{\varphi,t}(j) \mid j \in \Nbrs_{i}^{\type, \circ} (t, W)\} \subseteq K\) where $\circ \in \{\In, \Out\}$.

\mypara{Aggregation Monoid}
To combine neighbor values we introduce an aggregation space $(M,\boxplus,e)$, a commutative monoid. It serves as a domain for aggregating multisets of semantic values. Choices may include: $M = K$ with $\boxplus = \oplus$ (idempotent reduction),  $M = \mathbb{Z}$ with $\boxplus = +$ (counting),  $M = \mathbb{R}$ with $\boxplus = +$ (averaging), product spaces $M = K \times M'$ for hybrid semantics.

\noindent Let $\mathcal{A}_E : K^{*} \to M$ be a fold map that aggregates a multiset of neighbor values subject to the count constraint $E$, where $K^{*}$ denotes the set of finite multisets over $K$.
Then for each graph type $\type$ define: 
$\mathcal{A}_E\!\left(
  \{\, s_{\varphi,t}(j)
  \mid j \in \Nbrs_{i}^{\type,\circ}(t,W)\,\}
\right)
\;\in\; M$ with $\circ \in \{\In, \Out\}$.
This folding operation summarizes the multiset of neighbor values into a
single aggregate element of $M$. Different choices of $\mathcal{A}_E$
induce different interpretations of the count constraint.

\mypara{Readout} To obtain a value compatible with the temporal semantics, we map the
aggregated result back into $K$ via a readout map $h : M \to K$. This yields, for each graph type $\type$:
\(
h\!\left(
\mathcal{A}_E\!\left(
  \{\, s_{\varphi,t}(j)
  \mid j \in \Nbrs_{i}^{\type,\circ}(t,W)\,\}
\right)
\right)
\;\in\; K
\) with $\circ \in \{\In, \Out\}$.

\noindent \textit{Quantification over Graphs. }
Finally, we aggregate across graph types using semiring operations:
\(\mathcal{Q}_{\exists}^K = \bigoplus,
\text{ and } 
\mathcal{Q}_{\forall}^K = \bigotimes.
\) The graph operators are then defined as:
\(
    \robt{i}{\trace}{\InGO{\graph}{E}{W}{\#}\varphi}{t}
  =
  \mathcal{Q}_\#^K
  \Bigl(
    h\!\left(
      \mathcal{A}_E
      \bigl(
        \{\, s_{\varphi,t}(j)
        \mid j \in \Nbrs_{i}^{\type,\In }(t,W)\,\}\bigr)
    \right)
  \Bigr)_{\type\in\settypes}\)
and similarly for the $\Out$ operator.
For a fixed trace $\trace$, formula $\varphi$, time $t$, and quantifier 
$\# \in \{\exists, \forall\}$, we define the \emph{induced local graph operator}
$T_i^{\#} : S \to K$ by
\(
T_i^{\#}(s) \;=\; \mathcal{Q}_{\#}^{K}\!\Big(\, 
   h\big(\mathcal{A}_E(\{\, s(j) \mid j \in \Nbrs_i^{\type,\circ}(t, W) \,\})\big)
\,\Big)_{\type \in \settypes}.
\)
The graph operator semantics are then 
$\robt{i}{\trace}{\InGO{\graph}{E}{W}{\#}\varphi}{t} = T_i^{\#}(s_{\varphi,t})$, 
and analogously for $\Out$.

\noindent In this setting, graph aggregation is fully internal to the semiring,
and the resulting operator preserves the algebraic structure of $K$.

\noindent \textbf{Multi-agent aggregation and quantifiers.}
In addition to neighborhood-based aggregation, STL-GO supports
aggregation over the entire agent set. This yields the semantics of
multi-agent quantifiers. Let $s_{\varphi,t} \in S = K^{\nodes}$ denote the agent-indexed
valuation induced by a formula $\varphi$ at time $t$, i.e.,
\( 
s_{\varphi,t}(i) := \robt{i}{\trace}{\varphi}{t}.
\) We define a global aggregation map
$\mathsf{Acc} : S \to K$,
which combines the values across all agents.
The existential and universal quantifiers are obtained as special cases:
\(
\mathsf{Acc}_{\oplus}(s) = \bigoplus_{i \in \nodes} s(i),
\ 
\mathsf{Acc}_{\otimes}(s) = \bigotimes_{i \in \nodes} s(i).
\)
The semantics of multi-agent quantifiers are then: 
\(
  \robt{}{\trace}{\EX_{\nodes}\,\varphi}{t}
  = \mathsf{Acc}_{\oplus}(s_{\varphi,t}),
  \quad
  \robt{}{\trace}{\FA_{\nodes}\,\varphi}{t}
  = \mathsf{Acc}_{\otimes}(s_{\varphi,t}).
\)
More generally, $\mathsf{Acc}$ may be instantiated as an arbitrary
aggregation map, potentially factoring through an intermediate space
$M_A$ as in the graph-operator construction:
\(
S \;\to\; M_A \;\to\; K.
\)

\subsection{Spatio-temporal Algebras}
The construction of Section~\ref{subsec:graph-operators} separates
aggregation into an intermediate space $M$ and a readout map
$h : M \to K$. This enables the use of richer aggregation domains
that carry additional structural information beyond scalar robustness.
We now show how to construct such domains via product
algebras, yielding \emph{`spatio-temporal implicants'} that generalize
classical temporal semantics.

\mypara{Product aggregation spaces}
Let $(K,\oplus,\otimes)$ be the base semiring used for temporal
semantics. We extend the aggregation space by introducing an auxiliary
structure that tracks the \emph{origin} of satisfaction.
Let $R := \mathcal{P}(\nodes \times \timevar)$ where $\timevar$ denotes the time domain. Elements of $R$ represent sets of agent-time pairs witnessing the satisfaction of a formula. We define the product aggregation space $M := K \times R$ with componentwise aggregation $
(r_1, R_1) \;\boxplus\; (r_2, R_2)
\;:=\;
(r_1 \oplus r_2,\; R_1 \cup R_2)$.
Thus, $M$ simultaneously carries:
(i) a scalar semantic value $r \in K$, (ii) a  relational component $R \subseteq \nodes \times \timevar$.
    
\mypara{Spatio-temporal implicants}
In classical temporal logic, implicants correspond to time indices
at which a formula is satisfied. We extend this by interpreting the relational component $R$ as a
{spatio-temporal implicant}: a set that acts as a witness to the satisfaction of the formula. Our construction generalizes classic implicants by lifting them to the
spatio-temporal domain:
\(
R \subseteq \nodes \times \timevar,
\)
allowing formulas to track both \emph{which agents} and
\emph{when} they contribute to satisfaction. 
For example, for the ego agent $i$ fixed by the enclosing graph operator,
writing $r_j = \robt{j}{\trace}{\varphi}{t}$, we may define the fold as the 
ordered pair in $M = K \times R$:
\(
\mathcal{A}_E(\{r_j\}) = \Big(\, \bigoplus_j r_j,\; 
\{\, (j, t) \mid j \in \Nbrs_i^{\type,\circ}(t, W) \,\} \,\Big),
\)
where the first component is the scalar aggregate in $K$ and the second is 
the witness set in $R$. More generally, $R$ may track only those neighbors 
that contribute positively (e.g. $r_j \succ \bot_K$).

\mypara{Readout and Coupling}
The readout map $h : M \to K$ allows the relational component to influence the final semantic value.
A general form is: $h(r,R) = \Phi\bigl(r,\; \pi_{\nodes}(R)\bigr)$, where $\pi_{\nodes}(R) = \{\, i \mid \exists t,\ (i,t)\in R\,\}$
projects onto the set of participating agents, and $\Phi$ combines
signal-level and structural information.
For e.g., a hybrid robustness measure augments the scalar robustness $r$ with a count-based
term: $h(r,R) = r + \alpha\bigl(|\pi_{\nodes}(R)|\bigr),
$ where $\alpha \in \mathbb{R}$ is some constant.

\mypara{General Semiring Compositions}
More generally, the product construction extends to arbitrary families
of semirings:
\(
K = \prod_{\ell=1}^m K_\ell,
\ 
M = \prod_{\ell=1}^m M_\ell,
\)
with layer-specific aggregation and readout maps. This enables the
simultaneous propagation of multiple semantic quantities, such as: (a) robustness values, (b) counts or densities, (c) witness sets, (d) probabilistic weights.
Thus, STL-GO admits a unified algebraic semantics in which temporal,
spatial, and structural information are propagated compositionally
through product semiring constructions.
For some examples of algebras to pick for different semantic objectives, the
reader is encouraged to look at Appendix~\ref{sec:composition}.
\section{Structural Guarantees and Complexity}
We now characterize when the generalized aggregation framework described yields semantics that are
consistent with Boolean satisfaction.
\begin{definition}[Soundness and Completeness]
\label{def:sound-complete}
    We write $\trace,i,t \models \varphi$ to denote that agent $i$
satisfies $\varphi$ at time $t$ along trace $\trace$. In addition, we associate to each formula a quantitative semantics
(robustness value)
\(
\robt{i}{\trace}{\varphi}{t} \in K,
\)
which measures the degree of satisfaction of the formula in the
underlying algebra $K$. The quantitative semantics defined by the function $\robt{i}{\trace}{\varphi}{t} \in K$ is
\emph{sound} if
\(
\robt{i}{\trace}{\varphi}{t} \succ \bot_K
\;\Rightarrow\;
\trace,i,t \models \varphi,
\)
and \emph{complete} if
\(
\trace,i,t \models \varphi
\;\Rightarrow\;
\robt{i}{\trace}{\varphi}{t} \succeq \bot_K,
\) i.e., soundness ensures that positive robustness implies
satisfaction, while completeness ensures all satisfying
configurations are assigned non-negative robustness. Throughout we
adopt the standard STL convention in which the threshold value
$\bot_K$ itself counts as satisfying, so that the boundary case
$\robt{i}{\trace}{\varphi}{t} = \bot_K$ is treated as (weak)
satisfaction.
\end{definition}

\noindent The temporal semantics are sound and complete by classical
results~\cite{fainekos2009robustness}; our goal is to identify
conditions the graph-operator pipeline
$K^* \xrightarrow{\mathcal{A}_E} M \xrightarrow{h} K$
preserves this property.
Throughout, we assume that $K$ is equipped with a partial order
$\preceq$ compatible with the semiring operations, and write
$\bot_K$ for the threshold separating satisfaction from
violation (e.g., $0 \in \Rinf$ for the min-max algebra).

\noindent Before turning to soundness and completeness, we first record a weaker
but very general structural property of the layered semantics. Namely,
if each stage of the graph-aggregation pipeline is monotone w.r.t. the natural orders on the semantic and aggregation spaces, the overall operator preserves the pointwise ordering of agent-wise
valuations. This result is independent of any particular interpretation
of the threshold $\bot_K$ and applies to scalar, product, and
hybrid semantics.

\begin{definition}[Monotone map]
A function $f : X \to Y$ between partially ordered sets is
monotone if
\(x \preceq x' \;\Rightarrow\; f(x) \preceq f(x').\)
Monotonicity implies that increasing the robustness of inputs does not decrease the value of the output. 
\end{definition}

\begin{theorem}[Order preservation of layered aggregation]
\label{thm:order-preservation}
Let $(K,\preceq_K)$ be a partially ordered semiring, and let
\(S := K^{\nodes}
\)
be the induced semimodule of agent-indexed valuations equipped with the
pointwise order
\(
s \preceq_S s'
\ \Longleftrightarrow \ 
s(i) \preceq_K s'(i)
\;\; \forall i \in \nodes.
\)
Let $(M,\preceq_M)$ be a partially ordered aggregation space.
Fix an agent $i \in \nodes$, time $t$, graph type $\type \in \settypes$,
and consider the graph-aggregation pipeline
\(
S \xrightarrow{\mathrm{res}_i^\type} K^{\Nbrs_i^{\type,\circ}(t,W)}
\xrightarrow{\mathcal A_E} M
\xrightarrow{h} K,
\)
where $\mathrm{res}_i^\type$ denotes restriction of an agent-wise
valuation to the eligible neighborhood
$\Nbrs_i^{\type,\circ}(t,W)$.
Assume:
(1) the restriction map $\mathrm{res}_i^\type$ is order-preserving;
   (2) the fold map $\mathcal A_E$ is monotone with respect to the
    product order on $K^{\Nbrs_i^{\type,\circ}(t,W)}$;
    (3) the readout map $h : M \to K$ is monotone;
    (4) the graph-type quantifier $\mathcal Q_\#^K$ is monotone with
    respect to the pointwise order on $K^{|\settypes|}$.
Then the induced local graph operator
\(
T_i^\# : S \to K
\)
is monotone, i.e.,
\(
s \preceq_S s'
\ \Longrightarrow \ 
T_i^\#(s) \preceq_K T_i^\#(s').
\)
If this holds for every agent \(i \in \nodes\), then the induced global
operator
\(
T : S \to S,
\ 
(T(s))(i) := T_i^\#(s),
\)
is monotone with respect to the pointwise order on \(S\).
\label{thm:order}
\end{theorem}
\begin{proof}
Refer to Appendix~\ref{sec:proof-thm1} for the proof.
\end{proof}

\noindent Theorem~\ref{thm:order} is the structural basis for the results that follow.

\begin{proposition}[Soundness and Completeness of Graph Aggregation]
Let $(K,\oplus,\otimes,\ominus,\bot_K,\top_K)$ be a partially ordered
De Morgan algebra with order $\preceq$.
Let $(M,\boxplus,e)$ be a commutative monoid with order $\preceq_M$,
and let
\(\mathcal{A}_E : K^* \to M, \qquad
h : M \to K\)
define the aggregation pipeline for a count constraint $E = [e_1,e_2]$.
Define the composed map
\(
F_E := h \circ \mathcal{A}_E : K^* \to K.
\)
Assume:
(1)
For each input $r_j$,
\(
r_j \succ \bot_K \;\Longleftrightarrow\;
\text{neighbor } j \text{ satisfies } \varphi;
\)
(2) For every configuration $\{r_j\}$, if
$c^+ = |\{j : r_j \succ \bot_K\}|$,
$c^+ \in [e_1,e_2] \;\Rightarrow\; F_E(\{r_j\}) \succeq \bot_K$, and \
$c^+ \notin [e_1,e_2] \;\Rightarrow\; F_E(\{r_j\}) \prec \bot_K.$

\noindent Then the induced graph operator is sound and complete.
Consequently, since by Assumption~1
$c^{+} = \big|\{\, j : \trace, j, t \models \varphi \,\}\big|$,
the quantitative sign agrees with Boolean satisfaction:
\(
F_E(\{r_j\}) \succeq \bot_K
\;\Longleftrightarrow\;
(\mas, i, t) \models \InGO{\graph}{E}{W}{\#}\varphi ,
\)
where the threshold value $\bot_K$ is treated as satisfying, in keeping
with the boundary convention of Definition~\ref{def:sound-complete}.
\label{prop:graph-sound}
\end{proposition}
\begin{proof}
Refer to Appendix~\ref{sec:proof-prop-1} for the proof.
\end{proof}
\begin{lemma}[Quantifiers preserve threshold semantics]
\label{lem:quantifier-threshold}
Let $(K,\oplus,\otimes,\ominus,$ $\bot_K,\top_K)$ be an ordered
De~Morgan algebra, and let $\nodes$ be finite.
Assume that for all finite families $\{x_i\}_{i \in I} \subseteq K$,
\(
\bigoplus_{i \in I} x_i \succ \bot_K
\iff
\exists i \in I : x_i \succ \bot_K,
\text{ and }
\bigotimes_{i \in I} x_i \succ \bot_K
\iff
\forall i \in I : x_i \succ \bot_K.
\)
Then the multi-agent quantifiers satisfy:
\(
\robt{}{\trace}{\EX_{\nodes}\,\varphi}{t} \succ \bot_K
\iff
\exists i \in \nodes : \robt{i}{\trace}{\varphi}{t} \succ \bot_K,
\text{ and } \)
\(
\robt{}{\trace}{\FA_{\nodes}\,\varphi}{t} \succ \bot_K
\iff
\forall i \in \nodes : \robt{i}{\trace}{\varphi}{t} \succ \bot_K.
\)
Consequently, if the agent-local semantics of $\varphi$ are sound and
complete, then so are the multi-agent quantifiers
$\EX_{\nodes}\,\varphi$ and $\FA_{\nodes}\,\varphi$.%
\footnote{The strict separation ($\succ\bot_K$) used in the equivalences
above is consistent with the weak completeness convention of
Definition~\ref{def:sound-complete}: a boundary input $x_i=\bot_K$ leaves
both sides of the relevant equivalence false (for $\bigotimes=\min$) or is
absorbed by a strictly positive term (for $\bigoplus=\max$), so no
satisfying configuration is misclassified.}
\end{lemma}
\begin{proof}
Refer to Appendix~\ref{sec:proof-lem1} for the proof.
\end{proof}
\begin{theorem}[End-to-end soundness and completeness]
Assume that: (1) atomic predicate valuations are sound and complete, i.e., for every atomic predicate $\mu$,
    \(\robt{i}{\trace}{\mu}{t} \succ \bot_K
    \Leftrightarrow
    \trace,i,t \models \mu,
    \)
(2) the temporal/Boolean operators admit an interpretation over $(K,\oplus,\otimes,\ominus,\bot_K,\top_K)$ that is sound and complete
with respect to the Boolean STL semantics;
(3)  each graph operator satisfies Proposition~\ref{prop:graph-sound};
(4)  the multi-agent quantifiers satisfy
Lemma~\ref{lem:quantifier-threshold}, and the graph-type quantifiers
satisfy Proposition~\ref{prop:graph-sound};
then for every STL-GO formula $\phi$,
\(\robt{}{\trace}{\phi}{t} \succ \bot_K \Longleftrightarrow \trace,t \models \phi,
\)
\noindent that is, the quantitative semantics are sound and complete with respect
to the Boolean semantics.
\end{theorem}
\begin{proof}
    Refer to Appendix~\ref{sec:proof-thm2} for the proof.
\end{proof}
\noindent The preceding results establish soundness and completeness by verifying
that each syntactic construct (temporal and graph operators, and
quantifiers) preserves correspondence between quantitative
robustness and Boolean satisfaction. 
\\ 
\noindent\textbf{Monitoring Algorithm and Runtime Complexity.}
In Algorithm~\ref{alg:monitor} we describe a bottom-up procedure (with respect 
to the specification syntax) that computes the robustness of a given finite 
multi-agent trace, generalizing the Boolean STL-GO monitor of~\cite{stlgo}.
\renewcommand{\algorithmicrequire}{\textbf{Input:}}
\renewcommand{\algorithmicensure}{\textbf{Output:}}
\begin{algorithm}[t]
\caption{Quantitative monitoring of STL-GO}
\label{alg:monitor}
\begin{algorithmic}[1]
\Require Trace $\trace$; formula $\phi$; algebra $K$; accumulator $\mathcal{A}_E$, readout $h$.
\ForAll{subformulae $\psi$ of $\phi$, bottom-up; agents $i$; times $t$}
  \Switch{root operator of $\psi$}
    \Case{$\mu$ (predicate)} $\robt{i}{\trace}{\psi}{t}\gets \lbl(\Spmod_t,i,\mu)$ \EndCase
    \Case{$\neg,\wedge,\Untilbd{t_1}{t_2}$ (Boolean/temporal)} apply $\ominus,\otimes$, and the $\Untilbd{t_1}{t_2}$ recurrence over $K$ \EndCase
    \Case{$\EX_\nodes,\FA_\nodes$ (quantifier)} apply $\bigoplus_{i},\bigotimes_{i}$ over $K$ \EndCase
    \Case{$\InGO{\graph}{E}{W}{\#}\psi_1$ / $\OutGO{\graph}{E}{W}{\#}\psi_1$ (graph)}
      \State $\robt{i}{\trace}{\psi}{t}\gets \mathcal{Q}^K_{\#}\big(\,h(\mathcal{A}_E(\{\robt{j}{\trace}{\psi_1}{t}\mid j\in\Nbrs^{\type,\circ}_i(t,W)\}))\,\big)_{\type\in\settypes}$
    \EndCase
  \EndSwitch
\EndFor
\State \Return $\robt{}{\trace}{\phi}{\cdot}$
\end{algorithmic}
\end{algorithm}
\begin{theorem}[Monitoring complexity]\label{prop:complexity}
Let $\varphi$ be an STL-GO formula of size $|\varphi|$, evaluated over $N$ 
agents, horizon $T$, and $M$ graph types. Let
\(
\Delta := \max_{i,\,t,\,\type} \big|\Nbrs^{\type,\circ}_i(t,W)\big|
\)
denote the largest number of eligible neighbors any agent $i$ has in any graph type $\type$ at any 
time $t$. Assume that (i) all temporal operators use fixed-size bounded 
windows, (ii) the eligible-neighbor sets $\Nbrs^{\type,\circ}_i(t,W)$ are 
precomputed, and (iii) each accumulator update, readout, and semiring 
operation takes $O(1)$ time. Then Algorithm~\ref{alg:monitor} computes the 
robustness signal $\robt{}{\trace}{\varphi}{\cdot}$ in
\(O\!\big(|\varphi|\,N\,T\,M\,\Delta\big)\)
time.
\end{theorem}
\begin{proof}
     Refer to Appendix~\ref{proof:prop-complexity} for the proof.
     \end{proof}
\section{Experiments}
{\setlength{\textfloatsep}{4pt plus 1pt minus 1pt}%
\captionsetup{belowskip=0pt}%
\begin{figure*}[t]
\centering
\begin{subfigure}{0.4\textwidth}
    \centering
    \includegraphics[width=0.8\linewidth]{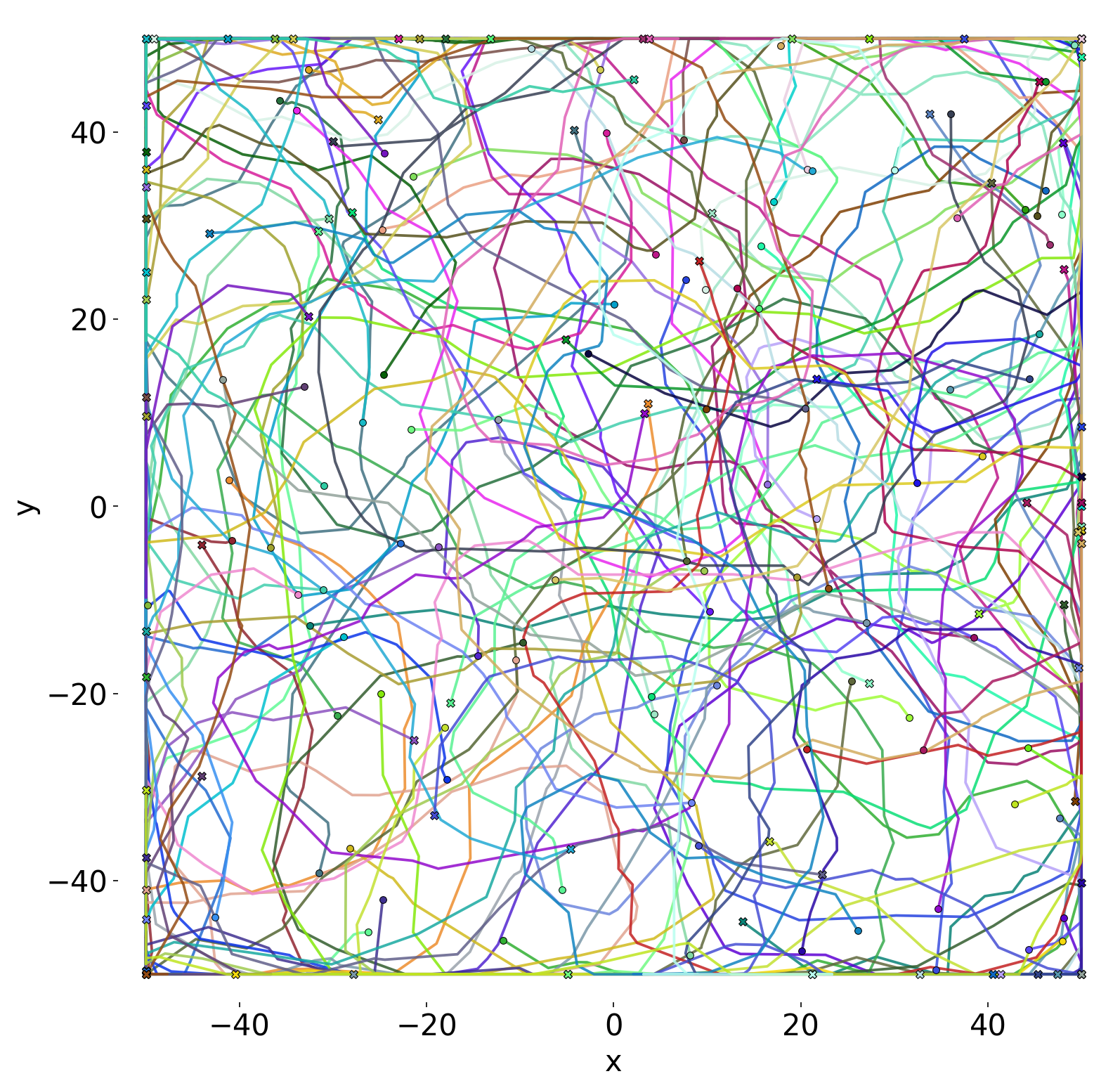}
    \caption{$|\nodes|=100$}
    \label{fig:trajectories-2D}
\end{subfigure}
\begin{subfigure}{0.4\textwidth}
    \centering
    \includegraphics[width=\linewidth]{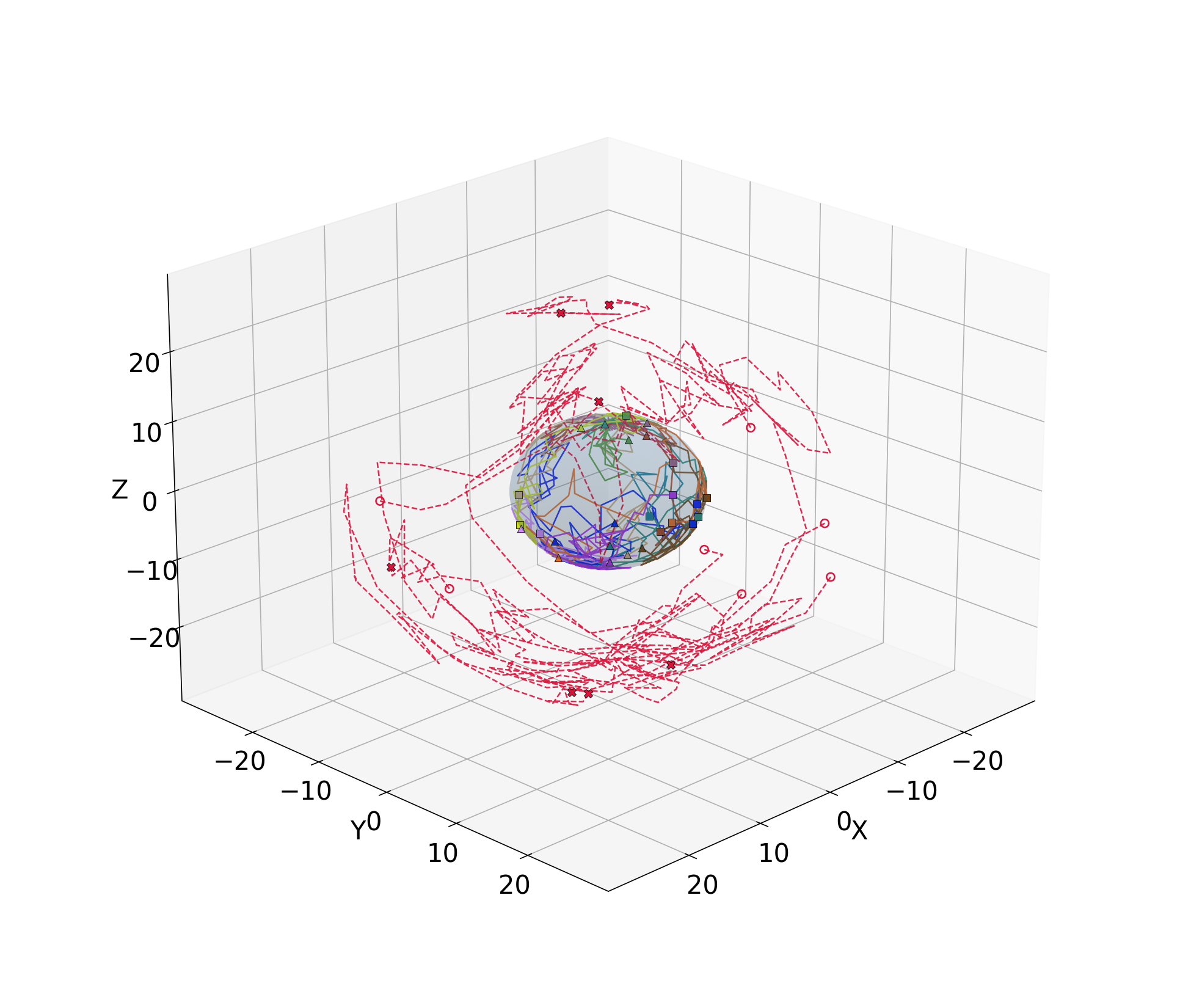}
    \caption{$|\nodes|=20, N_s = 14, N_f=6$}
    \label{fig:trajectories-3D}
\end{subfigure}
\caption{(a) Environment 1: 50x50 space with 100 agents, (b) Environment 2: 3D environment with 14 agents on the sphere, and 6 agents in space. }
\label{fig:trajectories}
\end{figure*}
We evaluate the STL-GO robustness semantics on two
multi-agent simulation environments with distinct dynamics
and interaction structures:
The two environments are chosen to stress different aspects of the semantics. \\
\noindent\textbf{Environment 1: 2D Bounded Region.}
\label{sec:env-2d}
\noindent We simulate $N \in \{10, 20, 50, 100\}$ agents over $T = 50$ discrete timesteps moving in a bounded planar region $[-50, 50]^2 \subset \mathbb{R}^2$.
Each agent's state is $\agentstate^i_t = (x^i_t, y^i_t,
\theta^i_t) \in \mathbb{R}^2 \times [0, 2\pi)$, comprising
position and heading angle.
The dynamics follow a stochastic Dubins-car model: 
\(
  x^i_{t+1} = x^i_t + v^i_t \cos\theta^i_{t+1} \Delta t, \ \
  y^i_{t+1} = y^i_t + v^i_t \sin\theta^i_{t+1} \Delta t, \ \
  \theta^i_{t+1} = (\theta^i_t + \omega^i_t \Delta t) \bmod 2\pi,
\) 
where $v^i_t \sim \mathrm{Unif}[0, 10]$ and 
$\omega^i_t \sim \mathrm{Unif}[-\pi/4, \pi/4]$
are linear and angular velocities.
(See Fig.~\ref{fig:trajectories-2D}).
The 2D Dubins-car region models a planar multi-robot team (e.g. ground 
search-and-rescue or distributed sensing), where the sensing/communication 
graphs are dense and time-varying, stressing the \emph{graph-operator} 
aggregation over large neighborhoods.

\noindent\textbf{Environment 2: 3D Sphere and Free Space.}
\label{sec:env-3d}
\noindent This environment is inspired by Earth-satellite systems, with two agent classes: ground stations constrained to a spherical surface and satellites moving freely in 3D space ($T=50$, Fig.~\ref{fig:trajectories-3D}).

\noindent\textit{Sphere-constrained agents ($N_s$).}
Each agent's state is parameterized by spherical angles $(\theta^i_t, \phi^i_t)$ on a sphere of radius $R_s$, with Cartesian position given by
\(
(x^i_t,\, y^i_t,\, z^i_t) = R_s\,(\sin\theta^i_t\cos\phi^i_t,\; \sin\theta^i_t\sin\phi^i_t,\; \cos\theta^i_t)
\).
The dynamics follow a random angular walk with $\Delta\theta^i_t, \Delta\phi^i_t \sim \mathrm{Unif}[-\pi/8, \pi/8]$ applied additively.

\noindent\textit{Free-space agents ($N_f$).}
These agents move freely in $\mathbb{R}^3$, initialized at some radius $R_f$ from the origin.
At each timestep, small perturbations are applied in spherical
coordinates:
$\Delta\phi^i_t, \Delta\theta^i_t \sim
\mathrm{Unif}[-\pi/100, \pi/100]$ and
$\Delta r^i_t \sim \mathrm{Unif}[-0.1, 0.1]$,
producing slow drifting trajectories near their initial
radii.
This system models a 
heterogeneous space-asset network (ground stations on a sphere, free-flying 
satellites), where connectivity is sparse and geometry-driven, stressing the 
\emph{multi-graph} quantification and edge-weight constraints spanning the dense/sparse and homogeneous/heterogeneous regimes that arise in 
real distributed CPS.

\noindent\textbf{Interaction Graphs.} 
At each timestep, three interaction graphs are computed
from the pairwise Euclidean distances
$d^t_{ij} = \|\pos^i_t - \pos^j_t\|_2$:
\emph{(i) Distance graph}
    $\mathcal{G}^{\mathrm{dist}}_t$: a weighted complete graph
    with edge weight $w^t_{ij} = d^t_{ij}$ for $i \neq j$.
\emph{(ii) Sensing graph}
    $\mathcal{G}^{\mathrm{sense}}_t$: an unweighted graph with
    edge $(i, j)$ iff $d^t_{ij} \leq R_{\mathrm{sense}}$ and $j$ is within the field of view of $i$,  i.e. the angle between the relative 
position $\pos^j_t - \pos^i_t$ and $i$'s heading vector 
$(\cos\theta^i_t, \sin\theta^i_t)$ is at most a half-angle $\beta$.
 \emph{(iii) Communication graph}
    $\mathcal{G}^{\mathrm{comm}}_t$: an unweighted graph with
    edge $(i, j)$ iff $d^t_{ij} \leq R_{\mathrm{comm}}$.
The topology of all the graphs changes at every timestep as agents move.

\noindent\textbf{Specifications.} 
We consider the following STL-GO specifications of increasing complexity, 
parameterized by a nesting depth $\ell \geq 0$.

\noindent \textit{Incoming reachability} ($\varphi^{r\text{-}in}_\ell$): The specifications $\varphi^{r\text{-}in}_\ell$ (Eq.~\ref{eq:spec-rec-in}) is defined by recursively nesting the $\In$ operator over the sensing graph $\mathcal{G}^{\mathrm{sense}}_t$, with edge-count bound $E=[1,4]$ and unconstrained weights $W=[-\infty,\infty]$.
It asserts that there exists an agent which, at some time within $[0,10]$, can be reached from a goal-satisfying agent via $\ell$ successive incoming edges over the sensing graph.
\begin{equation}
\begin{aligned}
\label{eq:spec-rec-in}
\psi_0 &= \InGO{\graph^{\mathrm{sense}}_t}{E}{W}{\exists} \, \mathtt{goal}, \;
\psi_{\ell + 1} = \InGO{\graph^{\mathrm{sense}}_t}{E}{W}{\exists} \, \psi_\ell, \;
\varphi^{r\text{-}in}_\ell = \EX_{\nodes} \, \Ev_{[0,10]} \, \psi_\ell, \; 
\end{aligned}
\end{equation}

\noindent \textit{Bi-directional reachability}($\varphi^{bi\text{-}dir}_\ell$): The formulas $\varphi^{bi\text{-}dir}_\ell$ (Eq.~\ref{eq:spec-rec-out-in}) necessitate bi-directional reachability over
$\mathcal{G}^{\mathrm{sense}}_t, \mathcal{G}^{\mathrm{comm}}_t$, with $E=[1,2]$ and $W=[-\infty,\infty]$.
The recursively defined subformulas $\psi^{\mathsf{out}}_\ell$ and $\psi^{\mathsf{in}}_\ell$ check reachability to a goal-satisfying agent via $\ell$ outgoing and incoming hops, respectively. 
Thus, $\varphi^{bi\text{-}dir}_\ell$ asserts that, within $[0,10]$, there exists an agent that can both reach and be reached by a goal-satisfying agent in $\ell$ hops.
\begin{equation}
\begin{aligned}
\label{eq:spec-rec-out-in}
\psi^{\mathsf{out}}_0 &= \OutGO{\{\mathcal{G}^{\mathrm{sense}}_t, \mathcal{G}^{\mathrm{comm}}_t \}}{E}{W}{\exists} \, \mathtt{goal}, \quad
\psi^{\mathsf{out}}_{\ell+1} = \OutGO{\{\mathcal{G}^{\mathrm{sense}}_t, \mathcal{G}^{\mathrm{comm}}_t\}}{E}{W}{\exists} \, \psi^{\mathsf{out}}_\ell \\
\psi^{\mathsf{in}}_0 &= \InGO{\{\mathcal{G}^{\mathrm{sense}}_t, \mathcal{G}^{\mathrm{comm}}_t\}}{E}{W}{\exists} \, \mathtt{goal}, \quad
\psi^{\mathsf{in}}_{\ell+1} = \InGO{\{\mathcal{G}^{\mathrm{sense}}_t, \mathcal{G}^{\mathrm{comm}}_t\}}{E}{W}{\exists} \, \psi^{\mathsf{in}}_\ell, \\
\varphi^{bi\text{-}dir}_\ell &= \EX_{\nodes} \, \Ev_{[0,10]} \left( \psi^{\mathsf{out}}_\ell \wedge \psi^{\mathsf{in}}_\ell \right), \quad \ell = 0, 1, \ldots
\end{aligned}
\end{equation}
\setlength{\intextsep}{0pt}        
\captionsetup{belowskip=0pt} 
 \begin{wrapfigure}{l}{0.52\textwidth}
    \centering    \includegraphics[width=\linewidth]{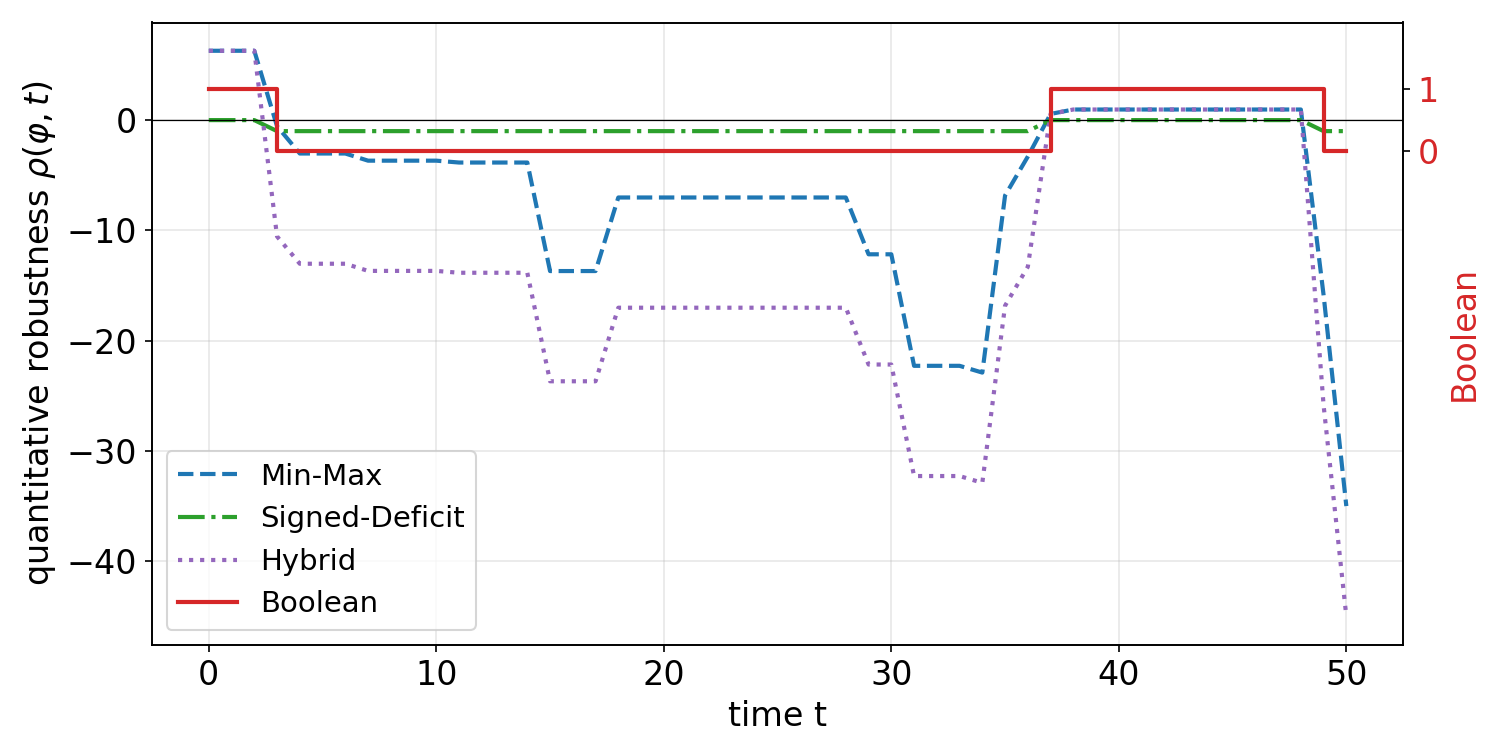}
    \caption{Robustness values across semantics for a sampled trajectory.}
    \label{fig:robust}
\end{wrapfigure}
\noindent \textit{Alternating Spatio-temporal operators} ($\varphi^{r\text{-}sp}_\ell$): These specifications (Eq. \ref{eq:spec-rec-alt})
interleave spatial and temporal operators in alternation over $\mathcal{G}^{\mathrm{dist}}_t, \mathcal{G}^{\mathrm{sense}}_t, \mathcal{G}^{\mathrm{comm}}_t$ graphs with edge-count bound $E=[1,20]$ and weight interval $W=[0, 20]$.
At each nesting level $\ell$, a temporal operator $\mathsf{T} = \{ \Ev, \Glob\}$ is applied first, followed by a spatial operator $\mathsf{S} = \{ \In, \Out \}$.
This allows the specification to encode diverse spatio-temporal patterns.
The atomic predicate $\mathtt{goal}$ holds at agent $i$ iff its position lies 
within the goal region $\mathcal{G}_{\mathrm{goal}} \subset \mathbb{R}^d$.
At each nesting level $\ell$, the temporal operator 
$\mathsf{T}_\ell \in \{\Ev, \Glob\}$ 
and the spatial operator 
$\mathsf{S}_\ell \in \{\InGO{\graph}{E}{W}{\exists}, \OutGO{\graph}{E}{W}{\exists}\}$ 
alternate, and the graph type is drawn cyclically from 
$\{\mathcal{G}^{\mathrm{dist}}_t, \mathcal{G}^{\mathrm{sense}}_t, 
\mathcal{G}^{\mathrm{comm}}_t\}$.
\begin{equation}
\label{eq:spec-rec-alt}
\psi^{sp}_0 = \mathtt{goal}, \quad 
\psi^{sp}_{\ell + 1} = \mathsf{T}_\ell \; \mathsf{S}_\ell \; \psi^{sp}_\ell, \quad 
\varphi^{r\text{-}sp}_\ell = \EX_{\nodes}\Ev_{[0,10]} \, \psi^{sp}_\ell, \quad
\ell = 0, 1, \ldots
\end{equation}
\noindent These formulae isolate the semantic target of the operators: 
$\varphi^{r\text{-}in}_\ell$ stresses nested \emph{incoming} 
count-constrained reachability; $\varphi^{bi\text{-}dir}_\ell$ adds 
multi-graph bidirectional reachability; and 
$\varphi^{r\text{-}sp}_\ell$ interleaves temporal and spatial 
operators to exercise layered composition. The recursive nesting 
depth $\ell$ controls the scaling complexity.

\noindent\textbf{Evaluation Metrics and Methodology.}
For each environment, we evaluate STL-GO specifications
under four semantic instantiations:
(i)~\emph{Boolean}: classical qualitative satisfaction
($K = \mathbb{B}$), yielding a binary verdict;
(ii)~\emph{Min-max}: standard STL robustness with the accumulator
$\mathcal{A}^{\mathrm{mm}}_{E} = \min(r_{(e_1)}, -r_{(e_2+1)})$;
(iii)~\emph{Signed-deficit}: the counting accumulator
$\mathcal{A}^{\mathrm{cd}}_E(\{r_j\}) = c-e_1$ if $c<e_1$,
$\min(c-e_1,\, e_2-c)$ if $e_1\le c\le e_2$, and $e_2-c$ if $c>e_2$,
where $c=|\{j: r_j\succeq\bot_K\}|$ counts satisfying neighbors;
(iv)~\emph{Hybrid}: the deficit-corrected accumulator
$\mathcal{A}^{\mathrm{h}}_{E} = r_{(k)} + \alpha \mathcal{A}^{\mathrm{cd}}_E(\{r_j\})$
(writing $k := e_1$ for the count threshold),
combining min-max and signed-deficit with
exchange-rate parameter $\alpha > 0$.
Signed-deficit counts $c=|\{j: r_j\succeq\bot_K\}|$ using the non-strict
threshold, so a nested subformula satisfied at zero margin is still counted,
and is sound and complete under the non-strict ($\succeq$) convention.
All share the same temporal De~Morgan algebra for temporal
operators and differ in the graph-operator accumulator.
\noindent We measure the wall-clock T(s) to evaluate the full STL-GO
specification for all agents at all timesteps, averaged over
100 independently generated trajectories and graph sequences.
To assess scalability, we vary two parameters independently:
(i)~the number of agents
$|\nodes| \in \{10, 20, 50, 100\}$ with fixed $T = 50$,
which affects neighbor set sizes and the cost of the
accumulator and multi-agent quantifiers; and
(ii)~the time horizon
$T \in \{10, 20, 50, 100\}$ with fixed $|\nodes| = 20$,
which affects the cost of temporal operators.
For each configuration, we report average monitoring time per
trajectory across 100 runs, broken down by semantic
instantiation.

\begin{table*}[t]
\centering
\caption{\footnotesize Results across $\varphi^{r\text{-}in}_\ell$, $\varphi^{bi\text{-}dir}_\ell$ (denoted $\varphi^{b\text{-}d}_\ell$), and $\varphi^{r\text{-}sp}_\ell$ over increasing recursive depths. For each, we report the percentage of satisfying trajectories (Sat.\%), the robustness range $[\min,\max]$ over the 100 trajectories, and computation time (T, seconds) for 100 trajectories, at fixed horizon $T=50$. For Boolean semantics this range lies in $\{0,1\}$, so $[0,1]$ indicates mixed verdicts and $[1,1]$ universal satisfaction. $\varphi^{r\text{-}in}_\ell$ and $\varphi^{b\text{-}d}_\ell$ use the 2D environment ($|\nodes|=100$); $\varphi^{r\text{-}sp}_\ell$ uses the 3D environment ($N_s=14$, $N_f=6$). Robustness and time are rounded to the nearest integer.
}
\label{tab:results_terminal}
\small
\setlength{\tabcolsep}{3pt}
\renewcommand{\arraystretch}{1.15}
\begin{tabular}{llllllllll}
\toprule
\multirow{2}{*}{Spec}
& \multirow{2}{*}{Sat.\%}
& \multicolumn{2}{c}{Bool.}
& \multicolumn{2}{c}{Min-Max}
& \multicolumn{2}{c}{Sign.-Def.}
& \multicolumn{2}{c}{Hybrid} \\
\cmidrule(lr){3-4}\cmidrule(lr){5-6}\cmidrule(lr){7-8}\cmidrule(lr){9-10}
&
& $[\min,\max]$ & T(s)
& $[\min,\max]$ & T(s)
& $[\min,\max]$ & T(s)
& $[\min,\max]$ & T(s) \\
\midrule
$\varphi^{r\text{-}in}_0$ & 94 & $[0,1]$ & 3 & $[-7,9]$  & 3 & $[-1,1]$ & 3 & $[-17,25]$ & 3 \\
$\varphi^{r\text{-}in}_1$ & 67 & $[0,1]$ & 3 & $[-16,9]$ & 3 & $[-1,1]$ & 3 & $[-36,19]$ & 3 \\
$\varphi^{r\text{-}in}_2$ & 55 & $[0,1]$ & 3 & $[-16,9]$ & 3 & $[-1,1]$ & 3 & $[-46,26]$ & 3 \\
$\varphi^{r\text{-}in}_3$ & 54 & $[0,1]$ & 3 & $[-16,9]$ & 3 & $[-1,1]$ & 3 & $[-56,28]$ & 3 \\
$\varphi^{r\text{-}in}_4$ & 54 & $[0,1]$ & 3 & $[-16,9]$ & 3 & $[-1,1]$ & 4 & $[-66,36]$ & 3 \\
\midrule
$\varphi^{b\text{-}d}_0$ & 100 & $[1,1]$ & 12 & $[3,10]$ & 11 & $[0,0]$ & 11 & $[4,20]$ & 11 \\
$\varphi^{b\text{-}d}_1$ & 100 & $[1,1]$ & 30 & $[3,10]$ & 28 & $[0,0]$ & 27 & $[7,30]$ & 27 \\
$\varphi^{b\text{-}d}_2$ & 100 & $[1,1]$ & 81 & $[3,10]$ & 73 & $[0,0]$ & 71 & $[8,40]$ & 76 \\
\midrule
$\varphi^{r\text{-}sp}_0$ & 65 & $[0,1]$ & 1 & $[-18,10]$ & 1 & $[-1,1]$ & 1 & $[-28,20]$ & 1 \\
$\varphi^{r\text{-}sp}_1$ & 10 & $[0,1]$ & 1 & $[-14,9]$  & 1 & $[-1,0]$ & 1 & $[-34,19]$ & 1 \\
$\varphi^{r\text{-}sp}_2$ & 26 & $[0,1]$ & 6 & $[-19,10]$ & 6 & $[-1,1]$ & 6 & $[-49,37]$ & 6 \\
$\varphi^{r\text{-}sp}_3$ & 6  & $[0,1]$ & 71 & $[-6,10]$  & 69 & $[-1,0]$ & 68 & $[-46,30]$ & 75 \\
\bottomrule
\end{tabular}
\end{table*}
\begin{table*}[t]
\centering

\begin{subtable}{0.48\textwidth}
\centering
\caption{Ablation on time horizon $T$.}
\label{tab:ablation_horizon}
\small
\setlength{\tabcolsep}{5pt}
\renewcommand{\arraystretch}{1.15}
\begin{tabular}{lrrrr}
\toprule
$T$ & Bool & MM & SD & Hyb. \\
\midrule
10  & 0.29 & 0.31 & 0.32 & 0.30 \\
20  & 0.62 & 0.71 & 0.73 & 0.74 \\
50  & 2.31 & 3.24 & 2.84 & 3.07 \\
100 & 6.70 & 10.65 & 9.58 & 9.68 \\
\bottomrule
\end{tabular}%
\end{subtable}
\hfill
\begin{subtable}{0.48\textwidth}
\centering
\caption{Ablation on number of agents.}
\label{tab:ablation_agents}
\small
\setlength{\tabcolsep}{5pt}
\renewcommand{\arraystretch}{1.15}
\begin{tabular}{lrrrr}
\toprule
$|\nodes|$ & Bool & MM & SD & Hyb. \\
\midrule
10  & 1.21  & 1.25  & 1.27  & 1.51 \\
20  & 2.16  & 3.09  & 2.86  & 2.87 \\
50  & 9.31  & 10.10 & 9.90  & 10.38 \\
100 & 25.35 & 26.68 & 27.61 & 26.19 \\
\bottomrule
\end{tabular}%
\end{subtable}
\caption{Ablation: wall-clock time T(s) for $\varphi^{r\text{-}in}_2$ over 100 trajectories, varying time horizon $T$ (left, fixed $|\nodes|=20$) and number of agents $|\nodes|$ (right, fixed $T=50$). Semantics: Boolean (Bool), Min-max (MM), Signed-Deficit (SD), Hybrid (Hyb.).
}
\label{tab:ablation_combined}
\end{table*}
\setlength{\textfloatsep}{4pt plus 1pt minus 1pt}   
\setlength{\intextsep}{4pt plus 1pt minus 1pt}  
\noindent \textbf{Results and Discussion.}
From Table~\ref{tab:results_terminal} we see that across all specifications the four semantics produce identical Boolean outcomes (Sat./Vio.), confirming that the quantitative semantics preserve correctness with respect to Boolean satisfaction; Fig.~\ref{fig:robust} visualizes this for a sampled trajectory.\footnote{Boolean is drawn as a post-step switching only at integer samples, while the quantitative curves are linearly interpolated; a Min-Max zero-crossing and the Boolean transition may thus appear shifted by up to one timestep, even though they agree in sign at each sampled $t$.} Runtimes are similar for the simpler $\varphi^{r\text{-}in}_{\ell}$, while differences grow with complexity in $\varphi^{bi\text{-}dir}_\ell$ and $\varphi^{r\text{-}sp}_{\ell}$, signed-deficit typically being fastest due to its simpler aggregation. The hybrid semantics is slightly more expensive but provides richer information: its larger range reflects the count-based term, which grows with the number of participating agents rather than being bounded by extrema. We emphasize that the four semantics compute different quantities and are not meant to compete on a single axis; the timing comparison quantifies the \emph{overhead} of the richer accumulators relative to the Boolean baseline, establishing acceptable cost.
 \noindent Table~\ref{tab:ablation_combined} reports the effect of time horizon $T$ and number of agents $|\nodes|$
on runtime for the spec $\varphi^{r\text{-}in}_2$. We fix $\varphi^{r\text{-}in}_2$ as the ablation baseline because it 
exercises nested graph operators without the large constant factors 
of the bi-directional family, making it ideal 
for isolating the effects of $T$ and $|\nodes|$.
Across $|\nodes|$, all four semantics scale near-identically, with differences within 
run-to-run variation. The clearest separation appears in the 
time horizon $T$, where the Boolean semantics scales best and the three 
quantitative accumulators track one another closely.


\noindent\textbf{Conclusions.}
\noindent We introduced a unified algebraic framework for quantitative semantics of STL-GO, separating temporal reasoning, graph aggregation, and multi-agent quantification. A key insight is that graph-based aggregation cannot, in general, be captured by standard semiring operations alone; instead, the structure of the aggregator determines whether the quantitative semantics align with Boolean satisfaction, and under appropriate conditions we obtain compositional soundness and completeness. Our experiments show that these semantics preserve Boolean correctness while incurring only modest overhead and scaling predictably with problem size. The proposed framework relies on the choice of aggregation operators,
which must be threshold-aligned with Boolean semantics. While we
identify suitable classes (e.g., order-statistic and hybrid), systematic
construction of expressive and well-behaved (e.g., smooth) aggregators
remains open. 
Future work includes developing efficient algorithms, designing
smooth aggregators, and extending the framework to stochastic
multi-agent settings.
\bibliographystyle{splncs04}
\bibliography{refs}

\clearpage
\appendix

\section*{\centering Appendix}

\section{Compositional Behavior of Algebraic Semantics}
\label{sec:composition}

\begin{table}[ht]
\centering
\caption{Algebraic instantiations of the aggregation pipeline for different semantics. (id = Identity). For the Hybrid row, the $\boxplus=(\max,+)$ column gives the monoid operations on $M=K\times\mathbb{Z}$; the first component is read out as the order statistic $r_{(k)}$ (the $k$-th largest value, $k:=e_1$), \emph{not} a plain $\max$-fold of the multiset.}
\label{tab:aggregation}
\small
\setlength{\tabcolsep}{3pt}
\renewcommand{\arraystretch}{1.2}

\begin{tabular}{@{}lcccccc@{}}
\toprule
\textbf{Semantics}
& $K$
& $(\oplus,\otimes)$
& $M$
& $\boxplus$
& $\mathcal{A}_E$
& $h$ \\
\midrule

Boolean
& $\Bool$
& $(\vee,\wedge)$
& $\Bool$
& $\vee$
& $\displaystyle
    \bigvee_{\substack{S \subseteq N_i \\ |S|=e_1}}
    \bigwedge_{j \in S} r_j$
& $\mathrm{id}$ \\

Min-max
& $\Rinf$
& $(\max,\min)$
& $\Rinf$
& $\max$
& $\min\bigl(r_{(e_1)}, -r_{(e_2+1)}\bigr)$
& $\mathrm{id}$ \\

Sign.-def.
& $\Rinf$
& $-$
& $\mathbb{Z}$
& $+$
& $
\begin{cases}
c - e_1, & c < e_1 \\
\min(c - e_1,\, e_2 - c), & e_1 \le c \le e_2 \\
e_2 - c, & c > e_2
\end{cases}
$
& $\mathrm{id}$ \\

Hybrid
& $\Rinf$
& $(\max,\min)$
& $K \times \mathbb{Z}$
& $(\max,\,+)$
& $\bigl(r_{(k)},\, \mathcal{A}^{\mathrm{cd}}_E\bigr)$
& $r_{(k)} + \alpha\,\mathcal{A}^{\mathrm{cd}}_E$ \\

\bottomrule
\end{tabular}
\end{table}


The algebraic semantics of STL-GO are parametric with respect to the
underlying algebraic structures used for temporal, graph-based, and
multi-agent aggregation. Different choices of these structures induce
distinct interpretations of satisfaction, enabling logical,
cost-based, and probabilistic reasoning within a unified framework.
We summarize the principal design choices below.

\subsection{Semiring choices.}
The temporal and multi-agent layers are governed by a De Morgan algebra
(or semiring) $(K,\oplus,\otimes,\ominus,\zero,\one)$.
Different instantiations yield different semantics:

\begin{itemize}
    \item \textbf{Lattice-based semirings.}
   $(K, \oplus, \otimes) \in
\{(\Bool, \vee, \wedge),\; (\Rinf, \max, \min)\}$.
Both operations are \emph{idempotent}: $a \oplus a = a$ and
$a \otimes a = a$, forming a bounded distributive lattice.
A De~Morgan negation exists ($\neg$ for Boolean, $-$ for min-max),
so the semantics of negation $\robt{i}{\trace}{\neg\varphi}{t}
= \ominus\,\robt{i}{\trace}{\varphi}{t}$ are well-defined without
restriction.
The Boolean algebra yields qualitative satisfaction; the min-max
algebra yields the standard STL robustness, where the value
measures the perturbation margin~\cite{fainekos2009robustness}.
\item \textbf{Tropical semirings.}
$(K,\oplus,\otimes)\!\in\!
\{(\Rpos \cup \{\infty\},\min,+),
(\mathbb{R} \cup \{-\infty\},\max,+)\}$.

Here $\oplus$ selects the best alternative while $\otimes = +$
{accumulates} values along a conjunction or temporal sequence,
rather than taking a bottleneck.
Multiplication is not idempotent ($a + a \neq a$), and no
De~Morgan negation exists: one cannot define $\ominus$ satisfying
$\ominus(a \oplus b) = \ominus a \otimes \ominus b$ when $\oplus$
and $\otimes$ have different algebraic structure.

    \item \textbf{Probabilistic and counting semirings.}
   $(\mathbb{R}_{\geq 0}, +, \times, 0, 1)$, or its restriction
to $[0,1]$.
Neither operation is idempotent: $p + p \neq p$ and
$p \times p \neq p$.
A negation $\ominus(p) = 1 - p$ can be defined on $[0,1]$, but it
satisfies De~Morgan's laws only under an independence assumption
($\Pr(\neg A) = 1 - \Pr(A)$ is exact, but
$\Pr(A \vee B) = \Pr(A) + \Pr(B)$ requires disjointness).
The interpretation is {probabilistic}: conjunction multiplies
probabilities and disjunction sums them, giving a satisfaction
likelihood under stochastic dynamics.
This family is relevant when agent behavior is modeled
stochastically and one seeks to evaluate $\Pr(\trace \models \phi)$
or an expected robustness.
\end{itemize}
Table~\ref{tab:semirings} summarizes the key properties.

\begin{table}[t]
\centering
\caption{Semiring families for the temporal and multi-agent layers.}
\label{tab:semirings}
\small
\renewcommand{\arraystretch}{1.2}
\begin{tabular}{@{}lccccl@{}}
  \toprule
  \textbf{Family}
    & $\oplus$/$\otimes$
    & \textbf{Idemp.}
    & \textbf{  De~Morgan.\ $\ominus$ }
    & \textbf{   Cont.   }
    & \textbf{Interpretation} \\
  \midrule
  Boolean
    & $\vee$ / $\wedge$
    & \checkmark
    & \checkmark
    & \ding{55}
    & qualitative satisfaction \\
  Min-max
    & $\max$ / $\min$
    & \checkmark
    & \checkmark
    & \checkmark
    & perturbation margin \\
  Tropical
    & $\min$ / $+$
    & $\oplus$ only
    & \ding{55}
    & \checkmark
    & cumulative cost \\
  Probabilistic
    & $+$ / $\times$
    & \ding{55}
    & partial
    & \checkmark
    & satisfaction likelihood \\
  \bottomrule
\end{tabular}
\end{table}

\subsection{Graph aggregation choices.}
The graph operators aggregate a multiset of neighbor robustness
values $\{r_1, \ldots,\\ r_m\} $$\subset K$ into a single value
$\mathcal{A}_E^{\K}(r_1, \ldots, r_m) \in K$, subject to the count
constraint $E = [e_1, e_2]$.  This aggregation is \emph{not} fixed
by the semiring; we identify some examples, distinguished
by the algebraic structure of their underlying aggregation.

\begin{itemize}
    \item \textbf{Idempotent reducers.}
[ $(K, \max, \min)$ or min-max selection].
Let $r_j = \robt{j}{\trace}{\varphi}{t}$
denote the robustness of $\varphi$ at neighbor $j$, and sort
the values in decreasing order:
$r_{(1)} \geq r_{(2)} \geq \cdots \geq r_{(m)}$,
with the convention $r_{(k)} = -\infty$ for $k > m$.
For a count interval $E = [e_1, e_2]$, define
\(
  \mathcal{A}_{[e_1,e_2]}^{\mathrm{os}}(r_1, \ldots, r_m)
  \;=\;
  \min\!\bigl(r_{(e_1)},\; -r_{(e_2+1)}\bigr).
\)
The first term $r_{(e_1)}$ is positive iff at least $e_1$
neighbors have positive robustness; the second term
$-r_{(e_2+1)}$ is positive iff at most $e_2$ do.
Their minimum is positive iff both conditions hold, i.e., iff
the count of satisfying neighbors lies in $[e_1, e_2]$.
The value measures the {perturbation margin of the critical
neighbor}: how much the $e_1$-th best neighbor's state can
change before the count drops below~$e_1$.
It is sound ( i.e. the sign agrees with Boolean verdict),
continuous in agent states.

    \item \textbf{Counting-based aggregators.}
[$(\mathbb{Z}, +, 0, -)$, an abelian group].
Let $c = |\{j : r_j \succeq \bot_K\}|$ and define
\[
  \mathcal{A}_{[e_1,e_2]}^{\mathrm{cd}}(r_1, \ldots, r_m)
  \;=\;
  \begin{cases}
    c - e_1 & c < e_1,\\
    \min(c - e_1,\; e_2 - c) & e_1 \le c \le e_2,\\
    e_2 - c & c > e_2.
  \end{cases}
\]
The value measures the {count surplus or deficit}: how many
neighbors could switch their Boolean verdict before the
cardinality constraint breaks.  A value of $+3$ means three
neighbors could switch their Boolean verdict before the
constraint breaks.  It is \emph{discontinuous}: a single
neighbor's robustness crossing zero induces a unit jump.

    \item \textbf{Averaging and linear aggregators.}
  [$(\mathbb{R}, +, \cdot\,,\, 0, 1)$ with
normalization by count]. We may define it as:
\(
  \mathcal{A}_{[e_1,e_2]}^{\mathrm{avg}}(r_1, \ldots, r_m)
  \;=\;
  \frac{1}{m}\sum_{j=1}^{m} r_j.
\)
This measures the {average quality} of neighbor satisfaction,
giving a smooth global summary that weights all neighbors equally
rather than focusing on the critical one.  It does not directly
enforce the count constraint $E$.
It is continuous and differentiable, but does not
recover the Boolean semantics in general (a positive average does
not imply that $e_1$ neighbors are individually satisfied).

    \item \textbf{Smooth or fuzzy aggregators.}
They include parameterized approximations, typically involving
$\exp$, $\log$, or sigmoid functions.
For example, a product-of-sigmoids
$\prod_j \sigma(\beta\, r_j)$ provides a smooth approximation to
counting. They are differentiable everywhere, making them amenable
to gradient-based optimization and learning. It has no clean algebraic
characterization; soundness holds only approximately, with the
approximation quality controlled by the temperature parameter~$\beta$.
\end{itemize}
Table~\ref{tab:accumulators} summarizes the key properties.

\begin{table}[t]
\centering
\caption{Accumulator families for the graph-operator layer.}
\label{tab:accumulators}
\small
\setlength{\tabcolsep}{5pt}
\renewcommand{\arraystretch}{1.2}
\begin{tabular}{@{}lcccl@{}}
  \toprule
  \multirow{2}{*}{\textbf{Family}}
    & \multirow{2}{*}{\textbf{Sound}}
    & \multirow{2}{*}{\textbf{Continuous}}
    & \textbf{Semiring} 
    & \multirow{2}{*}{\textbf{Interpretation}} \\
  &
  &
  &
  \textbf{-derived}
  & \\
  \midrule
  min-max
    & \checkmark & \checkmark & \checkmark
    & critical-neighbor margin \\
  Signed-deficit
    & \checkmark & \ding{55} & \ding{55}
    & count surplus/deficit \\
  Averaging
    & \ding{55}  & \checkmark & \ding{55}
    & mean neighbor quality \\
  Smooth/fuzzy
    & approx.    & \checkmark & \ding{55}
    & differentiable surrogate \\
  \bottomrule
\end{tabular}
\end{table}
\paragraph{Layer-wise composition.}
These algebraic choices may be applied uniformly across all layers, or
selected independently for temporal, graph, and multi-agent aggregation.
A uniform choice yields a consistent interpretation across time,
structure, and agents, while heterogeneous choices allow greater
flexibility, such as combining logical temporal reasoning with
count-based graph aggregation.
\section{Theoretical Results}
\subsection{Proof of Theorem 1}
\label{sec:proof-thm1}
\begin{proof}
Let $s,s' \in S$ with
\(
s \preceq_S s'.
\)
By definition of the pointwise order on $S$, this means
\(
s(j) \preceq_K s'(j)
\  \forall j \in \nodes.
\)

Fix an agent $i \in \nodes$ and a graph type $\type \in \settypes$.
Applying the restriction map to the eligible neighborhood
$N_i^{\type,\circ}(t,W)$ yields
\(
\mathrm{res}_i^\type(s)
\preceq
\mathrm{res}_i^\type(s')
\)
in the product order on $K^{N_i^{\type,\circ}(t,W)}$, since
restriction simply selects a subset of coordinates and is
order-preserving by assumption~(1).

By monotonicity of the fold map $\mathcal A_E$, we obtain
\[
\mathcal A_E\bigl(\mathrm{res}_i^\type(s)\bigr)
\preceq_M
\mathcal A_E\bigl(\mathrm{res}_i^\type(s')\bigr).
\]

Applying the monotone readout map $h : M \to K$ gives
\[
h\!\left(\mathcal A_E\bigl(\mathrm{res}_i^\type(s)\bigr)\right)
\preceq_K
h\!\left(\mathcal A_E\bigl(\mathrm{res}_i^\type(s')\bigr)\right).
\]

Thus, for every graph type $\type$, the corresponding per-type output is
order-preserved. Collecting these values over all graph types produces
two tuples in $K^{|\settypes|}$, say
\(
\bigl(x_\type\bigr)_{\type \in \settypes}
\preceq
\bigl(x'_\type\bigr)_{\type \in \settypes},
\)
again in the pointwise order. By monotonicity of the graph-type
quantifier $\mathcal Q_\#^K$, we conclude that
\[
T_i^\#(s)
=
\mathcal Q_\#^K\bigl((x_\type)_{\type \in \settypes}\bigr)
\preceq_K
\mathcal Q_\#^K\bigl((x'_\type)_{\type \in \settypes}\bigr)
=
T_i^\#(s').
\]
Hence the local graph operator is monotone.

If the same reasoning applies for every agent \(i \in \nodes\), then \( \forall i\),
\[
(T(s))(i)=T_i^\#(s) \preceq_K T_i^\#(s')=(T(s'))(i).
\]
Therefore
\(
T(s) \preceq_S T(s'),
\)
which proves monotonicity of the global operator \(T : S \to S\).
\end{proof}
\subsection{Proof of Proposition 1}
\label{sec:proof-prop-1}
\begin{proof}
Fix an agent $i$, time $t$, graph type $\type$, and let
\(
\mathbf r^\type
=
\bigl(\robt{j}{\trace}{\varphi}{t}\bigr)_{j \in \Nbrs_i^{\type,\circ}(t,W)}
\)
be the tuple of input values for that graph type.

By Condition (1),
\(
r_j^\type \succ \bot_K
\iff
\trace,j,t \models \varphi\). Hence, $c^+(\mathbf r^\type)
=
\bigl|\{j \in \Nbrs_i^{\type,\circ}(t,W) :
\trace,j,t \models \varphi\}\bigr|,$ i.e., the number of positive inputs is exactly the number of
satisfying neighbors.

Now apply Condition (2): $c^+(\mathbf r^\type) \in [e_1,e_2]
\iff
F_E(\mathbf r^\type) \succ \bot_K.$ Therefore, for each graph type $\type$, the quantitative output is
positive if and only if the Boolean count constraint is satisfied.

For existential graph-type quantification, the operator is $\mathcal{Q}_{\exists}^K = \bigoplus_{\type \in \settypes}.$ By sign-correctness of $\oplus$, $\bigoplus_{\type} x_\type \succ \bot_K
\iff
\exists \type \;:\; x_\type \succ \bot_K.$ Hence, $\robt{i}{\trace}{\InGO{\graph}{E}{W}{\exists}\varphi}{t} \succ \bot_K$ if and only if there exists a graph type $\type$ for which the
Boolean interval-count condition holds. This is exactly the Boolean
semantics of the existential graph-type quantifier.

We can similarly show the case for universal graph-type quantification using the $\bigotimes$ operator.
The argument for outgoing operators is identical, replacing
$\Nbrs_i^{\type,\mathrm{in}}$ by $\Nbrs_i^{\type,\mathrm{out}}$.
Thus the graph operators are sound and complete.
\end{proof}

\subsection{Proof for Lemma 1}
\label{sec:proof-lem1}
\begin{proof}
By definition of the multi-agent semantics,
\[
\robt{}{\trace}{\EX_{\nodes}\,\varphi}{t}
=
\bigoplus_{i \in \nodes}\robt{i}{\trace}{\varphi}{t},
\qquad
\robt{}{\trace}{\FA_{\nodes}\,\varphi}{t}
=
\bigotimes_{i \in \nodes}\robt{i}{\trace}{\varphi}{t}.
\]
Applying the assumed sign-correctness of finite $\oplus$- and
$\otimes$-aggregations yields
\[
\robt{}{\trace}{\EX_{\nodes}\,\varphi}{t} \succ \bot_K
\iff
\exists i \in \nodes:\;
\robt{i}{\trace}{\varphi}{t} \succ \bot_K,
\]
\[\robt{}{\trace}{\FA_{\nodes}\,\varphi}{t} \succ \bot_K
\iff
\forall i \in \nodes:\;
\robt{i}{\trace}{\varphi}{t} \succ \bot_K.
\]
If the local semantics of $\varphi$ are sound and complete, then
\[
\robt{i}{\trace}{\varphi}{t} \succ \bot_K
\Rightarrow
\trace,i,t \models \varphi,
\qquad
\trace,i,t \models \varphi
\Rightarrow
\robt{i}{\trace}{\varphi}{t} \succeq \bot_K.
\]
Combining these equivalences with the Boolean semantics of
$\EX_{\nodes}$ and $\FA_{\nodes}$ gives soundness and completeness of
the quantified formulas.
\end{proof}
\subsection{Proof for Theorem 2}\label{sec:proof-thm2}
\begin{proof}
    The proof follows by structural induction on the formula $\phi$.
\noindent\textit{Base case.}
If $\phi$ is an atomic predicate, the result holds by
assumption~(1).

\noindent\textit{Inductive step.}
Assume the claim holds for all immediate subformulas of $\phi$.

If $\phi$ is obtained from its subformulas using a Boolean or temporal
operator, then the result follows directly from assumption~(2), since
the temporal/Boolean fragment is sound and complete over $K$.

If $\phi$ is obtained using a graph operator, then by the inductive
hypothesis the inputs to the graph aggregation have the correct sign
with respect to the Boolean semantics of the subformula. Hence the
conditions of Proposition~\ref{prop:graph-sound} apply, and the graph
operator preserves soundness and completeness.

If $\phi$ is obtained by graph-type quantification or by multi-agent
quantification, the result follows from
Lemma~\ref{lem:quantifier-threshold}, since these quantifiers preserve
threshold semantics.
If $\phi$ is an agent embedding $i.\psi$, then
$
\robt{}{\trace}{i.\psi}{t}
=
\robt{i}{\trace}{\psi}{t}$,
and the result follows immediately from the inductive hypothesis.

Since all formation rules preserve soundness and completeness, the
result holds for all STL-GO formulas $\phi$.
\end{proof}
\subsection{Proof of Theorem~\ref{prop:complexity}}
\label{proof:prop-complexity}
    \begin{proof}
Algorithm~\ref{alg:monitor} performs a single bottom-up traversal of the parse
tree, visiting each of the $|\varphi|$ subformulae once, and for every
subformula computes robustness values for all $N$ agents and $T$ time steps. Predicate, Boolean, temporal, and quantifier nodes require only constant-time
algebraic operations per evaluation under the fixed-window assumption~(i),
yielding $O(NT)$ work per node. For a graph operator, each evaluation at agent $i$ and time $t$ iterates over
the $M$ graph types and folds the robustness values of the eligible neighbors
$\mathcal{N}^{\tau,\circ}_i(t,W)$, of which there are at most $\Delta$. By
assumptions~(ii) and~(iii), retrieving the neighbor set is free and each
accumulator update and readout is $O(1)$, so the fold costs $O(\Delta)$ per
graph type and $O(M\Delta)$ per $(i,t)$ pair. The graph-type quantifier
$Q^K_\#$ then combines the $M$ per-type results in $O(M)$, dominated by the
fold. Thus a graph-operator node costs $O(NTM\Delta)$.

\noindent Since graph-operator nodes dominate all other node types, summing over the
$|\varphi|$ subformulae gives total runtime $O(|\varphi|\,N\,T\,M\,\Delta)$.
\footnote{Since $\Delta \le N-1$ always, this gives the worst-case bound
$O(|\varphi|\,N^2\,T\,M)$; for sparse interaction graphs $\Delta \ll N$ and the cost is
substantially lower.}
\end{proof}
\subsection{General compositional soundness and completeness}
\begin{theorem}
\label{thm:general-compositional}
Let $K$ be a partially ordered set with threshold $\bot_K$, and define
the positive and non-positive regions
\(
K^+ := \{x \in K : x \succ \bot_K\},
\ \
K^- := \{x \in K : x \preceq \bot_K\}.
\)
Consider a compositional semantics in which each formula $\phi$ is
interpreted by a map
\(
\robt{}{\trace}{\phi}{t} \in K,
\)
constructed recursively from:
(i)atomic predicate valuations,
   (ii) a collection of operators
    \(
    F : K^n \to K
    \)
    applied to subformula values.
Assume:
\begin{enumerate}
    \item
    For every atomic predicate $\mu$,
    \(
    \robt{}{\trace}{\mu}{t} \in K^+
    \iff
    \trace,t \models \mu;
    \)

    \item
    for every semantic operator $F : K^n \to K$,
    there exists a Boolean operator
    \(
    f : \{0,1\}^n \to \{0,1\}
    \)
    such that for all inputs $x_1,\dots,x_n \in K$,
    \[
    F(x_1,\dots,x_n) \in K^+
    \iff
    f(\mathbf{1}_{x_1 \in K^+},\dots,\mathbf{1}_{x_n \in K^+}) = 1.
    \]
\end{enumerate}

Then for every formula $\phi$,
\(
\robt{}{\trace}{\phi}{t} \in K^+
\iff
\trace,t \models \phi.
\)
That is, the quantitative semantics are sound and complete with respect
to the Boolean semantics.
\end{theorem}

\begin{proof}
The proof is by structural induction on the formula $\phi$.

\noindent\textit{Base case.}
If $\phi$ is an atomic predicate, the claim holds by
assumption~(1).

\noindent\textit{Inductive step.}
Suppose $\phi = F(\phi_1,\dots,\phi_n)$ for some operator $F$.
By the induction hypothesis,
\(
\robt{}{\trace}{\phi_j}{t} \in K^+
\iff
\trace,t \models \phi_j
\quad \text{for all } j.
\)
Applying assumption~(2),
\[
\robt{}{\trace}{\phi}{t}
=
F\bigl(\robt{}{\trace}{\phi_1}{t},\dots,\robt{}{\trace}{\phi_n}{t}\bigr)
\in K^+
\]
if and only if
\(
f(\mathbf{1}_{\trace,t \models \phi_1},\dots,\mathbf{1}_{\trace,t \models \phi_n}) = 1,
\)
which is precisely the Boolean semantics of $\phi$.

Thus the claim holds for $\phi$, completing the induction.
\end{proof}
\section{Min-Max Robustness: Instantiation and Limitations}
\label{sec:minmax-robustness}

\noindent We now instantiate the three-tier algebraic semantics with the
min-max algebra $\MinMax$
and the min-max accumulator, giving concrete robustness
formulas for every operator in STL-GO.
We then examine a family of examples that expose a structural
limitation of the min-max accumulator and propose a
parameterized alternative.

\noindent \textit{Ego-centric STL-GO temporal operators.}
For agent $i$, trace $\trace$, and time $t$:
\begin{align}
  \robt{i}{\trace}{\mu}{t}
    &= \lbl(\Spmod_t, i, \mu),
  \notag\\
  \robt{i}{\trace}{\neg\varphi}{t}
    &= -\robt{i}{\trace}{\varphi}{t},
  \notag\\
  \robt{i}{\trace}{\varphi_1 \wedge \varphi_2}{t}
    &= \min\!\bigl(
         \robt{i}{\trace}{\varphi_1}{t},\;
         \robt{i}{\trace}{\varphi_2}{t}\bigr),
  \notag\\
  \robt{i}{\trace}{\varphi_1 \vee \varphi_2}{t}
    &= \max\!\bigl(
         \robt{i}{\trace}{\varphi_1}{t},\;
         \robt{i}{\trace}{\varphi_2}{t}\bigr),
  \notag\\
  \robt{i}{\trace}{\Evbd{t_1}{t_2}\,\varphi}{t}
    &= \max_{t' \in t+[t_1,t_2]}\;
       \robt{i}{\trace}{\varphi}{t'},
  \notag\\
  \robt{i}{\trace}{\Globbd{t_1}{t_2}\,\varphi}{t}
    &= \min_{t' \in t+[t_1,t_2]}\;
       \robt{i}{\trace}{\varphi}{t'},
  \notag\\
  \robt{i}{\trace}{\varphi_1\,\Untilbd{t_1}{t_2}\,\varphi_2}{t}
    &= \max_{t' \in t+[t_1,t_2]}\;
       \min\!\biggl(
         \robt{i}{\trace}{\varphi_2}{t'},\;
         \min_{t'' \in [t,\,t']}\,
         \robt{i}{\trace}{\varphi_1}{t''}
       \biggr).
  \label{eq:minmax-temporal}
\end{align}
The labelling function $\lbl(\Spmod_t, i, \mu)$ returns the
\emph{signed distance} of agent~$i$'s state from the predicate
boundary: positive if $\mu(\agentstate^i_t)$ holds, negative
otherwise, with magnitude proportional to the margin.

\noindent\textit{Graph operators.}
Fix agent $i$, time $t$, graph type $\type$, and let
$N = \Nbrs_i^{\type,\mathrm{in}}(t, W)$ be the set of qualifying
incoming neighbors (those connected to~$i$ by an edge with weight
in~$W$).
Let $r_j = \robt{j}{\trace}{\varphi}{t}$ for each $j \in N$, and
sort these values as $r_{(1)} \geq \cdots \geq r_{(|N|)}$, with
the convention $r_{(k)} = -\infty$ for $k > |N|$.
Then:
\begin{align}
  \robt{i}{\trace}{\InGO{\graph}{E}{W}{\exists}\varphi}{t}
  &= \max_{\type:\,\graph^{\type}\in\graphs_t}\;
     \min\!\bigl(r^{\type}_{(e_1)},\;
                 -r^{\type}_{(e_2+1)}\bigr),
  \label{eq:minmax-in-exist}\\[3pt]
  \robt{i}{\trace}{\InGO{\graph}{E}{W}{\forall}\varphi}{t}
  &= \min_{\type:\,\graph^{\type}\in\graphs_t}\;
     \min\!\bigl(r^{\type}_{(e_1)},\;
                 -r^{\type}_{(e_2+1)}\bigr).
  \label{eq:minmax-in-univ}
\end{align}
The outgoing operators
$\OutGO{\graph}{E}{W}{\#}$ are identical with
$\Nbrs_i^{\type,\mathrm{out}}$ in place of
$\Nbrs_i^{\type,\mathrm{in}}$.

\noindent \textit{Multi-agent operators.}
\begin{align}
  \robt{}{\trace}{\FA_\nodes\,\varphi}{t}
  &= \min_{i \in \nodes}\;
     \robt{i}{\trace}{\varphi}{t},
  &
  \robt{}{\trace}{\EX_\nodes\,\varphi}{t}
  &= \max_{i \in \nodes}\;
     \robt{i}{\trace}{\varphi}{t}.
  \label{eq:minmax-multiagent}
\end{align}
The remaining multi-agent temporal operators ($\neg$, $\wedge$,
$\vee$, $\Untilbd{t_1}{t_2}$) follow the same min/max structure
as the STL-GO operators.

\subsection*{Limitations}
The min-max accumulator reports a single value: the margin
of the $k$-th best neighbor, regardless of \emph{how many}
neighbors are satisfying or violating.  This makes it insensitive
to the \emph{count deficit}: configurations that differ
dramatically in the number of violating neighbors can receive
identical robustness.

\noindent\textit{Example.}Fix $|N| = 6$ neighbors and the threshold $k = e_1 = 5$
(``at least 5 neighbors satisfy~$\psi$'').
Consider the neighbor robustness multisets in
Table~\ref{tab:configurations}:
$S_1$ has all six neighbors satisfying robustly;
$S_2$ has exactly five at margin $+10$;
$S_3$ also has five satisfying, but the 5th-best has margin only
$+0.5$ (a much more fragile satisfaction);
$S_4$--$S_6$ have $4, 3, 2$ satisfying neighbors
respectively, each at $\pm 10$.

\begin{table}[t]
\caption{Neighbor robustness configurations
  ($k{=}5$, $|N|{=}6$).}
\label{tab:configurations}
\begin{minipage}[t]{0.52\textwidth}
\vspace{0pt}
\small
\renewcommand{\arraystretch}{1.15}
\begin{tabular}{@{}cl@{}}
  \toprule
  & \textbf{Neighbor robustness values} \\
  \midrule
  $S_1$ & $\{+10,\;+10,\;+10,\;+10,\;+10,\;+10\}$ \\
  $S_2$ & $\{+10,\;+10,\;+10,\;+10,\;+10,\;-10\}$ \\
  $S_3$ & $\{+10,\;+10,\;+10,\;+10,\;+0.5,\;-10\}$ \\
  $S_4$ & $\{+10,\;+10,\;+10,\;+10,\;-10,\;-10\}$ \\
  $S_5$ & $\{+10,\;+10,\;+10,\;-10,\;-10,\;-10\}$ \\
  $S_6$ & $\{+10,\;+10,\;-10,\;-10,\;-10,\;-10\}$ \\
  \bottomrule
\end{tabular}
\end{minipage}%
\hfill
\begin{minipage}[t]{0.44\textwidth}
\vspace{0pt}
\small
\textit{$S_1$--$S_3$ satisfy the constraint
($c^+ \geq 5$); $S_4$--$S_6$ violate it ($c^+ < 5$).
Ideally, the robustness should be:}
\begin{itemize}
  \item \textit{positive for $S_1$--$S_3$, with
    $S_1 > S_2 > S_3 > 0$;}
  \item \textit{negative and decreasingly so from $S_4$
    through~$S_6$, reflecting the growing count deficit.}
\end{itemize}
\end{minipage}
\end{table}
\noindent Under the (i) \textbf{min-max accumulator}, $S_1$ and
$S_2$ both receive $+10$: the accumulator cannot distinguish
having one spare neighbor from having zero.
$S_3$ receives $+0.5$, correctly reflecting its fragile 5th
neighbor, but $S_4$ through~$S_6$ all receive $-10$: a
configuration missing one neighbor is indistinguishable from one
missing five.
The order statistic captures \emph{how robustly} the critical
neighbor satisfies~$\psi$, but is entirely blind to
\emph{how many} neighbors are satisfying or missing.
Under the (ii) \textbf{signed-deficit accumulator} ($c^+ - k$),
the values +1, 0, 0, -1, -2, -3 provide a monotone
gradation reflecting the count surplus or deficit
(for this example $e_2 = \infty$, so $\mathcal{A}^{\mathrm{cd}}$
reduces to the single-sided $c - e_1$).
However, $S_2$ and $S_3$ both receive exactly~$0$ despite $S_2$
being far more robust than~$S_3$---the signed deficit discards
all signal-margin information.
We propose a (iii) \textbf{hybrid accumulator}:
$r_{(k)} + \alpha\,(c^+ - k)$, combining the order statistic
for signal-level margin with the scaled count deficit for
sensitivity to the number of satisfying neighbors.
The parameter~$\alpha > 0$ controls the \emph{exchange rate}
between signal margin (in robustness units) and count margin
(in number of agents).
As shown in Table~\ref{tab:accumulator-comparison}:
(a) $S_1 > S_2 > S_3 > 0$ for all $\alpha > 0$: the surplus
    neighbor in $S_1$ is reflected as a bonus of~$+\alpha$,
    and $S_2$ vs.\ $S_3$ are distinguished by the critical
    neighbor's margin ($+10$ vs.\ $+0.5$).
  (b) $S_4$ through $S_6$ are strictly decreasing, with 
    the spacing growing with $\alpha$: larger $\alpha$ amplifies
    the count-deficit penalty.
  (c) Within a fixed deficit level, the min-max term
    $r_{(k)}$ still distinguishes configurations by the margin
    of the critical neighbor.

\begin{table}[t]
\centering
\caption{Accumulator comparison
  ($k{=}5$, $|N|{=}6$).
  $\checkmark$: satisfies;
  $\times$: violates.}
\label{tab:accumulator-comparison}
\small
\renewcommand{\arraystretch}{1.15}
\begin{tabular}{@{}cc r c c w{r}{3em}w{r}{3em}w{r}{3em}@{}}
  \toprule
  \textbf{Set}
  & \textbf{Bool.}
  & $c^+$
  & \textbf{Min-max.}
  & \textbf{Sgn-def.}
  & \multicolumn{3}{c}{\textbf{Hybrid}\; $r_{(k)} + \alpha(c^+{-}k)$} \\
  \cmidrule(lr){4-4}
  \cmidrule(lr){5-5}
  \cmidrule(lr){6-8}
  & & & $r_{(5)}$ & $c^+{-}k$
  & $\alpha{=}1$ & $\alpha{=}5$ & $\alpha{=}10$ \\
  \midrule
  $S_1$ & $\checkmark$ & 6
    & $+10$  & $+1$            & $+11$   & $+15$   & $+20$ \\
  $S_2$ & $\checkmark$ & 5
    & $+10$  & $\phantom{+}0$  & $+10$   & $+10$   & $+10$ \\
  $S_3$ & $\checkmark$ & 5
    & $+0.5$ & $\phantom{+}0$  & $+0.5$  & $+0.5$  & $+0.5$ \\
  \midrule
  $S_4$ & $\times$ & 4
    & $-10$  & $-1$            & $-11$   & $-15$   & $-20$ \\
  $S_5$ & $\times$ & 3
    & $-10$  & $-2$            & $-12$   & $-20$   & $-30$ \\
  $S_6$ & $\times$ & 2
    & $-10$  & $-3$            & $-13$   & $-25$   & $-40$ \\
  \bottomrule
\end{tabular}
\end{table}
\end{document}